\documentclass[jcs]{iosart1c}     

\usepackage[T1]{fontenc}
\usepackage{times}%
\usepackage[numbers]{natbib}

\usepackage{fancyhdr}
\usepackage{latexsym}
\usepackage{graphicx} 
\usepackage{amsmath}
\usepackage{amssymb,latexsym,epsfig}
\usepackage{amsfonts}
\usepackage[latin1]{inputenc}
 \usepackage{textcomp}
 \usepackage{hyperref}

\newtheorem{exampleEM}{Example}

\newenvironment{example}{\begin{exampleEM}\em}{\qex\end{exampleEM}}
\newtheorem{theorem}{Theorem}
\newtheorem{definition}{Definition}
\newtheorem{lemma}{Lemma} 

\newenvironment{proof}{\noindent{\em Proof}.\ }{\vskip\parskip\noindent}

\usepackage{ifthen}

\newcommand{\qed}{\hfill$\Box$}
\newcommand{\qex}{\hfill$\blacksquare$}

\newcommand{\todo}[1]{\marginpar{\bf\large$[\ast]$}{\sf [#1]}}
\long\def\comment#1{}

\newcommand{\caI}{{\cal I}}

\newcommand{\caP}{{\cal P}}
\newcommand{\caR}{{\cal R}}

\newcommand{\caT}{{\cal T}}
\newcommand{\caX}{{\cal X}}

\newcommand{\sort}[1]{\ensuremath{\mathsf{#1}}}

\newcommand{\occ}[1]{\mathit{Pos}(#1)}
\newcommand{\funocc}[1]{\mathit{Pos}_{\Symbols}(#1)}

\newcommand{\var}{{\cal V}ar}

\newcommand{\vars}[1]{{\cal V}ar(#1)}

\newcommand{\Variables}{\caX}
\newcommand{\Symbols}{\Sigma}

\newcommand{\subterm}[2]{#1|_{#2}} 
\newcommand{\replace}[3]{#1[#3]_{#2}}
\newcommand{\domain}[1]{\mathit{Dom}(#1)}
\newcommand{\range}[1]{\intrvar{#1}}
\newcommand{\intrvar}[1]{\mathit{Ran}(#1)}
\newcommand{\idsubst}{\textit{id}}

\newcommand{\TermsOn}[5]{{\caT^{#4}_{#1}(#2)}{}_{#3}^{#5}}
\newcommand{\Terms}{\TermsOn{\Symbols}{\Variables}{}{}{}}
\newcommand{\ETerms}{\TermsOn{\Symbols\!/\!E}{\Variables}{}{}{}}
\newcommand{\TermsS}[1]{\TermsOn{\Symbols}{\Variables}{\sort{#1}}{}{}}
\newcommand{\ETermsS}[1]{\TermsOn{\Symbols\!/\!E}{\Variables}{\sort{#1}}{}{}}

\newcommand{\GTermsOn}[2]{\caT^{#2}_{#1}}
\newcommand{\GTerms}{\GTermsOn{\Symbols}{}}

\newcommand{\GTermsS}[1]{\GTermsOn{\Symbols,\sort{#1}}{}}

\newcommand{\TermsSt}{\TermsS{State}}

\newcommand{\composeRel}{;}
\newcommand{\rewrite}[1]{\rightarrow_{#1}}

\newcommand{\rewritePos}[2]{\stackrel{#2}{\rightarrow}_{#1}}

\newcommand{\narrow}[2]{\stackrel{#1}{\leadsto}_{#2}}

\newcommand{\csu}[3]{\textit{CSU}_{#3}({#1})}
\newcommand{\csuV}[3]{\textit{CSU\/}^{#2}_{#3}({#1})}

\newcommand{\restrict}[1]{{|_{#1}}}

\newcommand{\congr}[1]{=_{#1}}

\newcommand{\nI}[1]{\ensuremath{#1{\notin}\caI}}
\newcommand{\inI}[1]{\ensuremath{#1{\in}\caI}}

\newcommand{\Pone}{\caP_1}
\newcommand{\Ptwo}{\caP_2}  
 \newcommand{\RTcomp}{\caR_{\caP_{1} \, ;_S \caP_{2}}}
  \newcommand{\RTdc}{\caR_{synch(\caP_{1} \, ;_S \caP_{2})}}

 \newcommand{\SeqComp}{\caP_{1} \, ;_S \caP_{2}}

\pubyear{0000}
\volume{0}
\firstpage{1}
\lastpage{1}

\begin{document}

\begin{frontmatter}    

\title{Effective Sequential Protocol Composition in Maude-NPA}
\runningtitle{Effective Sequential Protocol Composition in Maude-NPA}

\maketitle


\author[A,C]{\fnms{Sonia} \snm{Santiago}},
\author[A]{\fnms{Santiago} \snm{Escobar}},
\author[B]{\fnms{Catherine} \snm{Meadows}},
\author[C]{\fnms{Jos\'e} \snm{Meseguer}}
\runningauthor{S. Santiago et al.}
\address[A]{
DSIC-ELP, Universitat Polit\`ecnica de Val\`encia, 
Camino de Vera, s/n
46022 Valencia
Spain\\
E-mail: \{ssantiago,sescobar\}@dsic.upv.es
}
\address[B]{Naval Research Laboratory, 
4555 Overlook Ave SW, Washington, DC 20375, USA\\
 E-mail:  meadows@itd.nrl.navy.mil}
\address[C]{ 
Department of Computer Science, University of Illinois at Urbana-Champaign, 
The Thomas M. Siebel Center for Computer Science
201 N. Goodwin Ave.
Urbana, IL 61801-2302, USA\\
E-mail: meseguer@cs.uiuc.edu
}

\begin{abstract}
 Protocols do not work alone, but together, one protocol relying on another
 to provide needed services.  Many of the problems in cryptographic protocols
 arise when such composition is done incorrectly or is not well understood.
 In this paper we discuss an extension to 
 the Maude-NPA syntax and its operational semantics
to support dynamic sequential composition of protocols, so that protocols can be 
specified separately
and composed when desired.  This allows one to reason about many
 different compositions
with minimal changes to the specification, as well as improving, in terms of both performance and ease of specification, on an earlier composition extension we presented in \citep{EMMS10}. We show how compositions can be defined and executed
symbolically in Maude-NPA using the compositional syntax and semantics.
We also provide an experimental analysis of the performance of Maude-NPA using the compositional syntax and semantics,
and compare it to the performance of a syntax and semantics for composition developed in earlier research.
Finally, in the conclusion we give some lessons learned about the best ways of extending narrowing-based state reachability tools, as well as comparison with related work and future plans.   
\end{abstract}


\end{frontmatter}






\section{Introduction}\label{sec:intro}

The area of formal analysis of cryptographic protocols has been an active one since the mid 1980's.
The idea is to verify 
protocols that use cryptography to guarantee security against an attacker
---commonly called the {\em Dolev-Yao} attacker \citep{dolev-yao}--- 
who has complete
control of the network, and can intercept, alter, and redirect traffic, create new traffic on his/her own, perform all
operations available to legitimate participants, and may have access to some subset of the longterm keys
of legitimate principals.    Whatever approach is taken, the use of formal methods has had a long
history, not only 
for providing formal proofs of security, but  also for uncovering bugs 
and security
flaws that in some cases had remained unknown 
long after the original protocol's publication.

A number of approaches have been taken to the formal verification of cryptographic protocols.
One of the most popular is model checking, in which the interaction of the protocol with the
attacker is symbolically executed.  Indeed, model-checking of secrecy (and later, authentication) in protocols in the bounded-session
model (where a {\em session} is a single  execution of a process representing an honest principal) has been shown
to be decidable \citep{rusinowitch-np}, 
and a number of bounded-session model checkers exist.  Moreover, a number
of unbounded model checkers either make use of abstraction to enforce
decidability, or allow for the possibility of non-termination.

It is well known that many problems in the security of cryptographic
protocols arise when the protocols are composed. This is true whether the composition is parallel, in which two different protocols are executed in an interleaved fashion, or sequential, in which one or more child protocols use information from executing a parent protocol.  Protocols that work  correctly
in one environment may fail when they are composed with new protocols in new environments,
either because the properties they guarantee are not quite appropriate for the new environment,
or because the composition itself is mishandled.   Security of parallel composition can generally be achieved by avoiding ambiguity about which  protocol a message belongs to (as in, e.g. \citep{GuttmanThayer00,CiobacaC10}).  The necessary conditions for security of sequential composition are  harder to pin down, since they depend on the guarantees offered and needed by the particular protocols being analyzed.

To see an example of the problems that can arise, consider the analyses of the Internet Engineering Task Force's (IETF) Group Domain of Interpretation (GDOI) protocol \cite{baugher2003group},
in which the third author of this paper was involved.  GDOI is a secure multicast protocol built on top of the IKE Version 1 (V1) \cite{ikev1}
key distribution protocol, which had already undergone at least one formal analysis \cite{meadows1999analysis} and substantial peer review by the IETF.  However, early versions of GDOI had two subtle flaws, one arising from the composition of GDOI with IKE, and the other arising from the way a subprotocol of GDOI was used by the parent protocol.  One was a type confusion attack that took advantage of the fact that IKE V1 headers began with random numbers instead of a field indicating what type of protocol it was \cite{meadows2004formal}.  An intruder could take advantage of this confusion  to obtain a group key to which it was not entitled.  The discovery of this attack led to a redesign of the
GDOI protocol before it was submitted as a standard.   Another attack involved a subprotocol of GDOI, called the \emph{Proof of Possession} (POP) protocol.  The GDOI specification was not clear about the situations in  which POP was to be used.  Once these were clarified, it was discovered the POP was also subject to an attack \cite{meadows2004deriving}.  This discovery led to a modification of GDOI to fix this vulnerability.

The importance of understanding sequential composition has long been acknowledged, and there are a number
of logical systems that support such compositional reasoning.  The Protocol Composition Logic (PCL) begun with 
\citep{durgin01csfw} is probably the first protocol
logic to approach composition in a systematic way. Logics such as the Protocol Derivation Logic (PDL) \citep{cerv05csfw}, and tools such as
the Protocol Derivation Assistant (PDA)  \citep{anlauff06} and the Cryptographic Protocol Shape Analyzer (CPSA)   \citep{doghmi} also support reasoning about composition.
All of these are logical systems and tools that support reasoning about the properties guaranteed by the protocols.  One uses the logic to determine whether the properties guaranteed by the protocols are adequate.  This is a natural way to approach sequential composition, since one can use these tools to determine whether the properties guaranteed by one protocol are adequate for the needs of another protocol that relies upon it.  Thus,   PCL   and  the authentication tests methodology
underlying CPSA   are used to analyze key exchange standards and electronic commerce protocols in 
 \citep{datta03mfps} and  in \citep{guttman02csfw} 
respectively, via sequential composition out of simpler components.

Less attention has been given to handling sequential composition when model checking protocols, especially in the case in which an instance of a parent protocol can spawn multiple instances of  subprotocols, e.g., in the case of a master key agreement protocol that can be used multiple times to generate a session key.  We believe that this is an imbalance that needs to be corrected, for  logical systems and state exploration techniques make complementary contributions to our understanding of the security of a protocol.  Logical methods allow us to construct proofs from basic assumptions, much as we develop protocols that use basic cryptographic algorithms.  These logical systems provide insight into  how a protocol achieves security, and what basic assumptions it depends on.  State exploration tools, on the other hand, provide concrete attacks that can be used in fixing a protocol.  Moreover, they are also useful for discovering behaviors that, while they may not violate specified security properties, nevertheless turn out to be undesirable.  This can be used to inform and refine the logical systems.  Finally, state-exploration-based models can provide a useful semantics for logical systems.


The problem is in providing a specification and verification environment that supports composition.  
This is not necessarily straightforward; we note in particular that the
two leading formal calculi underlying most current cryptographic protocol analysis tools, the pi calculus \citep{AbadiF01}, and strand  spaces \citep{strands} 
do not provide general sequential  operators that can be used to specify  protocol compositions.  
That does not mean of course that one cannot define in these languages protocols that are actual compositions of other protocols; 
it just means that sequential protocol composition is not supported at the language level and must be encoded by the user in ways 
that may depend on the particular composition at hand.  
The aim of this work is to provide  
specification primitives for a wide variety of compositions at the language level in a way 
 that is both transparent to the user and sound and complete with respect to a desired semantics.   

There are several ways that composition can be handled in state exploration systems.   One is to not  modify the tool at all, but to handle everything at the specification level, by concatenating protocols that are being composed: for example, a master key agreement protocol followed by a session key
distribution protocol.  This requires no modification, but besides being tedious to specify when many different ways of composition are possible, it also cannot be used to represent cases in which a parent protocol can have an arbitrary number of children, as is indeed the case in the master/session key case.  Another is to compose protocols at execution time, but without modifying the operational semantics of the tool, an approach we took in \citep{EMMS10}. 
Although this minimizes the modifications made to the tool, it can lead to counterintuitive and inefficient methods of specification and analysis, since we are adapting composition to the original semantics instead of the other way around.  Finally, we can extend the operational semantics of the tool, but minimize 
such an extension by adding or modifying as few semantic rules as possible.

In this paper we describe how  this third approach has been applied to the Maude-NPA protocol analysis tool and its strand-space-based semantics.  We first give an abstract composition semantics, first introduced in \citep{EMMS10}, that extends Maude-NPA's operational semantics using the concept of \emph{parameterized} strands  \citep{guttman2001security} augmented with a separate composition operator.  We then describe an extension of Maude-NPA to Maude-NPA \emph{with composition} via \emph{synchronization messages}, in which composition is achieved by means of strand space parameters, but without the separate composition operator, and prove soundness and completeness of the operational semantics of Maude-NPA with composition with respect  to a subset of the abstract semantics. 
This extension allows us to minimize the changes made to Maude-NPA, as well as the number of rewrite rules that need to be added to its operational semantics.  We provide evidence that this approach to extending Maude-NPA is probably optimal by comparing it with an earlier approach we took  in \citep{EMMS10}, in which composition was implemented via \emph{protocol transformation}, in which the synchronization  messages implementing composition were passed along the Dolev-Yao channel.  Although this required fewer modifications to Maude-NPA, since it already supported  communication in the Dolev-Yao model, the additional communication overhead had a negative impact on performance.  We illustrate this via experiments comparing the performance of  both composition approaches, via synchronization messages and protocol transformation.

Thus the contributions of this paper are the following:

\begin{enumerate}
\item It provides a formal definition of sequential protocol composition in the strand space model  (Section~\ref{sec:mnpa-comp}).
\item It provides a new operational semantics of protocol composition in Maude-NPA (Section~\ref{sec:mnpa-comp}).
\item It provides a simple and intuitive syntax for protocol composition in Maude-NPA
(Section~\ref{sec:mnpa-comp}).
\item It describes an implementation of protocol composition directly in Maude-NPA via the operational semantics, giving a proof of its soundness and completeness with respect to a subset of the abstract semantics.
(Section~\ref{sec:directComposition}). 
\item It provides an experimental evaluation of the new operational semantics, and
compares its performance with respect to the protocol transformation technique presented in~\citep{EMMS10}
(Section~\ref{sec:experiments-composition}).
\end{enumerate}


The rest of the paper is organized as follows.  
In Section~\ref{sec:examples-composition} 
we introduce  two 
motivating examples of sequential protocol composition, which will be used
throughout this paper as running examples.
After some 
preliminaries in Section~\ref{sec:preliminaries},
 we give an overview of the Maude-NPA tool and its operational semantics   in Section~\ref{sec:maude-npa} (referred to as \emph{basic Maude-NPA} to distinguish it from Maude-NPA with composition).
In Section~\ref{sec:mnpa-comp} we describe the syntax for sequential protocol composition and its operational semantics.
In Section~\ref{sec:directComposition} we describe  an implementation of composition via synchronization messages, and show that it is sound and complete with respect to
the semantics given in Section~\ref{sec:mnpa-comp}.  
A presentation   of  the protocol transformation approach to protocol composition
is described in Section~\ref{sec:transfComposition}, preparatory to an experimental evaluation and performance comparison between the two approaches   given in Section~\ref{sec:experiments-composition}.
Finally, in Section~\ref{sec:conc} we conclude the paper and discuss related and future work, as well as lessons learned.

%

\section{Motivating Examples}\label{sec:examples-composition}

  In this section we provide several motivating examples of sequential composition.
  These  examples give a flavor for the variants of sequential composition that are used in constructing cryptographic protocols. 
 A single parent protocol instance can be composed with either many instances of a child protocol, 
 or with only one such child instance.
 Likewise, parent protocol roles can determine child protocol roles, or child protocol roles can be unconstrained.  
 In Section~\ref{sec:examples-NSLDB}
we provide  an example of a one-parent, one-child protocol composition,
 which  appeared in \citep{Guttman-location}
 and which is subject to a distance hijacking attack previously described in \cite{EMMS10}; we also provide a corrected version that it is proved to be secure against distance hijacking.
 In Section~\ref{sec:examples-NSLKD}
 we provide  an example of a one-parent, many-children protocol composition which is proved secure by our tool.

 \subsection{NSL Distance Bounding Protocol}\label{sec:examples-NSLDB}
 
 In this example of a one-parent, one-child protocol composition, appeared in   \citep{Guttman-location},
the participants  first use NSL to agree on a secret nonce. 
We reproduce the NSL protocol below using textbook Alice-and-Bob notation where
$A \rightarrow B: m$ means participant with name $A$ sending the message $m$ to the participant with name $B$:

\vspace{1ex}
\begin{quote}
\begin{enumerate}
\item $A \rightarrow B: \{N_A,A\}_{pub(B)}$
\item $B \rightarrow A: \{N_A,N_B,B\}_{pub(A)}$
\item $A \rightarrow B: \{N_B\}_{pub(B)}$
\newcounter{enumi_saved}
\setcounter{enumi_saved}{\value{enumi}}
\end{enumerate}
\end{quote}
\vspace{1ex}
\noindent where 
$\{M\}_{pub(A)}$ means message $M$ encrypted using the public key of principal with name $A$, $N_A$ and $N_B$ are nonces generated by the respective principals, and we use the comma as message concatenation.

 The agreed nonce $N_{A}$ is then used in a
distance bounding protocol described below.  This is a type of protocol, originally proposed by 
\citep{desmedt}
for smart cards, which has received new interest in recent years for its possible application in 
wireless environments \citep{capkun}. 
%
The idea behind the protocol is that
Bob uses the round trip time of a challenge-response protocol with Alice to compute an upper bound on her distance from him 
according to the following
protocol:

\begin{quote}
\begin{enumerate}
\setcounter{enumi}{\value{enumi_saved}}
\item $B \rightarrow A: N_B'$

\noindent Bob records the time at which he sent $N_B'$

\item $A \rightarrow B: N_A \oplus N_B'$

\noindent Bob records the time he receives the response and checks 
the equivalence $N_{A} = N_A \oplus N_B' \oplus N'_{B}$.  If this holds, he uses the round-trip time of his challenge and response to estimate his distance from Alice
\end{enumerate}
\end{quote}

\noindent
where $\oplus$ is the exclusive-or operator satisfying 
associativity (i.e., $X \oplus (Y \oplus Z) = (X \oplus Y) \oplus Z$)
and 
commutativity (i.e., $X \oplus Y = Y \oplus X$)
plus the self-cancellation
property
$X \oplus X = 0$ and 
the identity property $X \oplus 0 = X$.
Note that Bob is the initiator 
and Alice is the responder of the distance bounding protocol,
in contrast to the NSL protocol.

This protocol must satisfy two requirements.  
The first is that it must guarantee that $N_A \oplus N'_B$ 
 was sent after $N'_B$ was
received, or Alice will be able to pretend that she is closer than she is.  Note that if Alice and Bob do not agree on $N_A$
beforehand, then Alice will be able to mount the following attack:
%
%
$B \rightarrow A: N_B'$ 
and then
$A \rightarrow B: N$.
%
%
\noindent Of course, $N = N_B' \oplus X$ for some $X$.  But Bob has no way of telling if Alice computed $N$ using
$N_B'$ and $X$, or if she just sent a random $N$.
Using NSL to agree on a $X = N_A$ in advance prevents this type of attack.

Bob also needs to know that the response comes from whom it is supposed to be from.  In particular, an attacker should not be able to impersonate Alice.  Using NSL to agree on $N_A$ guarantees that only Alice and Bob can know $N_A$, so the attacker cannot impersonate Alice.  However, it should also be the case that an attacker cannot pass off Alice's response as his own. 
This is not the case for the NSL distance bounding protocol, which is subject to a form of what has come
to be known as the Distance Hijacking Attack \citep{DBLP:conf/ndss/CremersRC12}
\footnote{This is not meant as a denigration of \citep{Guttman-location}, whose main focus is on timing models in strand spaces, not the design of distance bounding protocols.}.   This attack was found by the authors of this paper by inspection and has been previously described in \cite{EMMS10}.

\begin{enumerate}\label{attack-nsldb}
\item[a)] Intruder $I$ runs an instance of NSL with Alice as the initiator and $I$ as the responder, obtaining
a nonce $N_A$.
\item[b)] $I$ then runs an instance of NSL with Bob with $I$ as the initiator and Bob as the responder, using $N_A$ as the initiator nonce.
\item[c)] $B \rightarrow I : N_B'$
where
$I$ does not respond, but Alice, seeing this, thinks it is for her.
\item[d)] $A \rightarrow I : N_B' \oplus N_A$
where 
Bob, seeing this thinks this is $I$'s response.
\end{enumerate}

If Alice is closer to Bob than $I$ is, then $I$ can use this attack to appear closer to Bob than he is.
This attack is a textbook example of a composition failure.  NSL has all the properties of a good key distribution protocol,
but fails to provide all the guarantees that are needed by the distance bounding protocol.  However, in this case we can fix the problem, not by changing NSL, but by changing the distance bounding protocol so that it provides a stronger guarantee:
\begin{quote}
\begin{enumerate}
\setcounter{enumi}{\value{enumi_saved}}
\item $B \rightarrow A: N_B'$
\item $A \rightarrow B: h(N_A , A) \oplus N_B'$
\noindent where $h$ is a collision-resistant hash function.%
\end{enumerate}
\end{quote}
\noindent
As we show in our analysis in Section~\ref{sec:experiments-composition}, this prevents the attack.  $I$ cannot pass off Alice's nonce as his own because it is now bound to her name.  

The distance bounding example is a case of a one parent, one child protocol composition.  Each instance of the parent NSL
protocol can have only one child distance bounding protocol, since the distance bounding protocol depends upon the assumption that $N_A$ is known only by $A$ and $B$.  But since the distance bounding protocol reveals $N_A$, it cannot be used with the same $N_A$ more than once.

\subsection{NSL Key Distribution Protocol}\label{sec:examples-NSLKD}

Our next example is a one parent, many children protocol composition, also using NSL.  This type of composition arises, for example, in key distribution protocols in which the parent protocol is used to generate a master key, and the child protocol is used to generate a session key.  In this case, one wants to be able to run an arbitrary number
of instances of the child protocol with the same master key.

In the distance bounding example the initiator of the distance bounding protocol was always the child of the
responder of the NSL protocol and vice versa.  In the key distribution example, the initiator of the session key protocol
can be the child of either the initiator or the responder of the NSL protocol.  So, we have two possible child executions after NSL:

\vspace{2ex}
\begin{tabular}{@{}l@{}l@{}}
\begin{minipage}{.5\linewidth}
\begin{enumerate}
\setcounter{enumi}{\value{enumi_saved}}
\item $A \rightarrow B: \{Sk_A\}_{h(N_A,N_B)}$
\item $B \rightarrow A: \{Sk_A ; N'_{B}\}_{h(N_A,N_B)}$
\item $A\rightarrow B: \{N'_{B}\}_{h(N_A,N_B)}$
\end{enumerate}
\end{minipage}


&

\begin{minipage}{.5\linewidth}
\begin{enumerate}
\setcounter{enumi}{\value{enumi_saved}}
\item $B \rightarrow A: \{Sk_B\}_{h(N_A,N_B)}$
\item $A \rightarrow B: \{Sk_B ; N'_{A}\}_{h(N_A,N_B)}$
\item $B \rightarrow A: \{N'_{A}\}_{h(N_A,N_B)}$
\end{enumerate}
\end{minipage}

\end{tabular}
\vspace{1ex}

\noindent
where 
$Sk_{A}$ is the session key generated by principal $A$
and $h$ is again a collision-resistant hash function.
This protocol is proved secure by our tool in Section~\ref{sec:experiments-composition}.

\section{Background on Term Rewriting}
\label{sec:preliminaries}

In this section we provide background on the concepts from term rewriting
used in this paper.  Due to space constraints, this section is rather terse and mainly intended to 
reference purposes.  The reader should consult it as needed. Readers familiar with such terminology and notation can skip this section and proceed to 
the next section, where we provide examples of protocol specification.

We follow the classical notation and terminology from
\citep{Terese03} for term rewriting
and from
\citep{Meseguer92-tcs,tarquinia} for rewriting logic and order-sorted notions.

We assume an \textit{order-sorted signature} $\Symbols$
with a finite poset of sorts $(\sort{S},\leq)$ and a finite 
number of function symbols.
We assume an $\sort{S}$-sorted family 
$\Variables=\{\Variables_\sort{s}\}_{\sort{s} \in \sort{S}}$
of disjoint variable sets with each $\Variables_\sort{s}$
countably infinite.
$\TermsS{s}$
denotes the set of terms of sort \sort{s},
and
$\GTermsS{s}$ the set of ground terms of sort \sort{s}.
We write 
$\Terms$ and $\GTerms$ for the corresponding term algebras.
We write $\var(t)$ for the set of variables present in a term $t$.
The set of positions of a term $t$ is written $\occ{t}$,
and
the set of non-variable positions $\funocc{t}$.
The subterm of $t$ 
at position $p$
is $\subterm{t}{p}$, and $\replace{t}{p}{u}$ is 
the result of replacing $\subterm{t}{p}$ by $u$ in $t$.
In Maude-NPA, we use sorts to indicate such things as which terms
are intended to be parts of messages, and which terms, such as strands,
are part of the higher-level infrastructure.  We also use sorts to provide restrictions
on how messages may be constructed.  For example, we can specify an encryption function
as symbol $e$ of arity two, where the first argument must be of sort $\sort{key}$, while the second
argument is of sort $\sort{message}$, where $\sort{key} < \sort{message}$.                                                  

A \textit{substitution} $\sigma$ is a sort-preserving mapping
from a finite subset of $\Variables$ 
to $\Terms$.
The set of variables assigned by $\sigma$ is $\domain{\sigma}$
and
the set of variables
introduced by $\sigma$ is $\range{\sigma}$.
The identity
substitution is $\idsubst$.
Substitutions are homomorphically extended to $\Terms$.
Application of substitution $\sigma$ to term $t$ is denoted by $t\sigma$.
Thus $e(K,X) (\sigma = \{ K \mapsto key(A,B), X \mapsto n(A,r) \} = e(key(A,B),n(A,r))$.
The restriction of $\sigma$ to a set of variables $V$ is 
${\sigma}\restrict{V}$.
The composition of two substitutions is 
$x(\sigma\theta)=(x\sigma)\theta$ for $x\in\Variables$.

A \textit{$\Symbols$-equation} is an unoriented pair $t = t'$,
where $t \in \TermsS{s}$,
$t' \in \TermsS{s'}$,
and $s$ and $s'$ are sorts in the same connected
component of the poset  $(\sort{S},\leq)$.
Given a set $E$ of $\Symbols$-equations,
order-sorted equational logic
induces 
a congruence relation $\congr{E}$ on terms $t,t' \in \Terms$;
see \citep{tarquinia}.   Throughout this
paper we assume that $\GTermsS{s}\neq\emptyset$ for every sort \sort{s}.
We denote the $E$-equivalence class of a term $t\in\Terms$
as $[t]_E$ and the $S$-sorted families of sets of $E$-equivalence classes of all terms $\Terms$
and $\TermsS{s}$
as $\ETerms$, and $\ETermsS{s}$ for the quotient set of sort \sort{s}, respectively.
A substitution $\sigma$ is \emph{more general modulo $E$} than another substitution $\theta$,
written $\sigma \sqsupseteq_{E}  \theta$,
iff there is a substitution $\rho$
such that $\sigma\rho \congr{E} \theta$,
i.e., such that $x\sigma\rho =_E x\theta$ for each $x\in\Variables$.  
In Maude-NPA we use equations to represent the properties
of crypto systems.  Thus, if we want to represent the fact that decryption with a key cancels out encryption with
the same key, we can use the equation $d(K,e(K,X)) = X$.  


For a set $E$ of $\Symbols$-equations, an \textit{$E$-unifier} 
for a $\Symbols$-equation $t = t'$ is a
substitution $\sigma$ s.t. $t\sigma \congr{E} t'\sigma$. 
For
$\vars{t}\cup\vars{t'} \subseteq W$, a set of substitutions $\csuV{t =
  t'}{W}{E}$ is said to be a \textit{complete} set of 
   of $E$-unifiers of an equation
$t = t'$  away from $W$ iff:
(i) each $\sigma \in
\csuV{t = t'}{W}{E}$ is an $E$-unifier of $t = t'$;
(ii) for
any $E$-unifier $\rho$ of $t = t'$ there is a $\sigma \in
\csuV{t=t'}{W}{E}$ such that $\subterm{\sigma}{W} \sqsupseteq_{E} \subterm{\rho}{W}$; 
(iii) for all
$\sigma \in \csuV{t=t'}{W}{E}$, $\domain{\sigma} \subseteq
(\vars{t}\cup\vars{t'})$ and $\range{\sigma} \cap W = \emptyset$.
If the set of variables $W$ is irrelevant or is understood from the context,
we write $\csu{t = t'}{W}{E}$ instead of $\csuV{t = t'}{W}{E}$.
We say 
that $E$-unification is \emph{finitary}  if
$\csu{t = t'}{W}{E}$  contains a finite
number of $E$-unifiers for any equation $t = t'$, and \emph{unitary} if it contains most one.
For example,  $E$-unification when $E = \{d(K,e(K,X)) = X\}$ is finitary but not unitary.  
For example, the complete set of unifiers  $\csu{d(K,X) = Y}{W}{E}$
contains two substitutions: $\sigma_1= \{Y \mapsto  d(K,X)\}$ and  $\sigma_2 = \{X \mapsto e(K,Y)\}$ .

A \textit{rewrite rule} is an oriented pair $l \to r$, where
$l \not\in \Variables$
and
$l,r \in \TermsS{s}$ for some sort $\sort{s}\in\sort{S}$. 
An \textit{(unconditional)
  order-sorted rewrite theory} is a triple $\caR = (\Symbols,E,R)$
with $\Symbols$ an order-sorted signature, $E$ a set of
$\Symbols$-equations, and $R$ a set of rewrite rules.  
A \emph{topmost rewrite theory} $(\Symbols,E,R)$ is a rewrite theory 
s.t.
for each $l \to r \in R$, $l,r\in\TermsSt$ for a top sort \sort{State},
and no operator in $\Symbols$ has \sort{State} as an argument sort.  In Maude-NPA, topmost rewriting is used to describe which
states can follow from other states.  That is, these rewrite rules are topmost rules of form $S \to S'$, where $S$ and $S'$ are both of topmost terms sort \sort{State}. The theory $E$ used by Maude-NPA describes the equational properties of the cryptosystem.

The rewriting relation $\rewrite{R}$ on
$\Terms$ is 
$t \rewritePos{R}{p} t'$ 
(or $\rewrite{R}$) 
if 
$p \in \funocc{t}$,
$l \to r\in R$, 
$\subterm{t}{p} = l\sigma$, 
and $t' =
\replace{t}{p}{r\sigma}$
for some $\sigma$.
The relation $\rewrite{R/E}$
on $\Terms$ is 
$\congr{E} \composeRel\rewrite{R}\composeRel\congr{E}$,
i.e.,
$t \rewrite{R/E} s$
iff
$\exists u_1,u_2\in\Terms$ s.t.
$t \congr{E} u_1 \rewrite{R} u_2 \congr{E} s$.
Note that
$\rewrite{R/E}$ on $\Terms$
induces a relation 
$\rewrite{R/E}$ on $\ETerms$
by
$[t]_{E} \rewrite{R/E} [t']_{E}$ iff $t \rewrite{R/E} t'$.  
%
%
The relation $\rewrite{R/E}$
is undecidable in general, since $E$-congruence classes can be arbitrarily large,
and 
the simpler relation $\rewrite{R,E}$ is used.
The rewriting relation $\rewrite{R,E}$ on
$\Terms$ is performed by applying narrowing to representatives of 
${t \rewritePos{R,E}{p} t'}$ 
(or $\rewrite{R,E}$) 
if 
$p \in \funocc{t}$,
$l \to r\in R$, 
$\subterm{t}{p} \congr{E} l\sigma$, 
and $t' =
\replace{t}{p}{r\sigma}$
for some $\sigma$.  
The narrowing relation $\narrow{}{R}$ on
$\Terms$ is 
$t \narrow{p}{\sigma,R} t'$ 
(or $\narrow{}{\sigma,R}$, $\narrow{}{R}$) 
if 
$p \in \funocc{t}$,
$l \to r\in R$, 
$\sigma \in \csu{\subterm{t}{p} = l}{W}{\emptyset}$, 
and $t' = \sigma(\replace{t}{p}{r})$.  
Assuming that $E$ has a finitary and complete unification algorithm,
the narrowing relation $\narrow{}{R,E}$ on
$\Terms$ is
$t \narrow{p}{\sigma,R,E} t'$ 
(or $\narrow{}{\sigma,R,E}$, $\narrow{}{R,E}$) 
if 
$p \in \funocc{t}$,
$l \to r\in R$, 
$\sigma\in\csu{t|_p = l}{V}{E}$, 
and $t' = (\replace{t}{p}{r})\sigma$.

Maude-NPA uses narrowing modulo $E$ to perform state space exploration.     Its use of topmost rewrite theories provides several advantages; see \citep{narrowing-hosc06}:
(i) 
the relation 
$\rewrite{R,E}$ achieves the same effect as the relation $\rewrite{R/E}$,
and
(ii) we obtain a completeness result between narrowing ($\narrow{}{R,E}$) and 
rewriting ($\rewrite{R/E}$),
in the sense that a reachability problem has a solution
iff narrowing can find an instance of it.
%

For equational theories $E$ describing the properties of cryptosystem, Maude-NPA uses
$E=E' \uplus Ax$ such that the   equations $E'$ oriented as left-to-right rules are
confluent, coherent, and 
terminating modulo 
axioms $Ax$ such as commutativity ($C$), 
associativity-commutativity ($AC$),
or associativity-commutativity plus identity ($ACU$)
of some function symbols.
We also require axioms $Ax$ to be \emph{regular}, i.e.,
for each equation $l = r \in Ax$, $\var(l) = \var(r)$.

Note that axioms such as
commutativity ($C$), 
associativity-commutativity ($AC$),
or associativity-commutativity plus identity ($ACU$) are regular.
Maude-NPA has  both dedicated and generic algorithms 
for solving unification problems in such theories $E' \uplus Ax$
under appropriate conditions
\citep{Escobar-JLAP}.


\section{Basic Maude-NPA's Execution Model and Protocol Analysis}\label{sec:maude-npa}  

 In this section we describe the core syntax and semantics of Maude-NPA as described in  \citep{FOSAD07}.
 We refer to it here as \emph{basic Maude-NPA} to distinguish it from Maude-NPA with composition.  When we describe features
 that will be modified once composition is added, we refer explicitly to ``basic Maude-NPA''.  When a feature is the same for
 both versions we simply say ``Maude-NPA.''

In Maude-NPA the behaviors of protocols is modeled using  rewrite theories.  Briefly, a protocol $\caP$ is a set of strands.  Each strand is either a \emph{protocol strand} that describes the actions of a role played by an honest principal, or an \emph{intruder strand} describing the ways in which the intruder can derive new messages, e.g. by generating nonces or applying functions symbols to messages it already knows.  Although the two are conceptually different, the are processed the same way in basic Maude-NPA.  Thus, given a protocol $\caP$, its behavior in basic Maude-NPA
is modeled by the rewrite theory $(\Symbols_\caP,E_\caP, \caR_\caP)$,
where $\Sigma_{\mathcal{P}}$ is the signature
defining the sorts and function symbols for the cryptographic functions and for
all the state constructor symbols,
$E_{\mathcal{P}}$ is a set of equations specifying the
\emph{algebraic properties} of the cryptographic functions and the state constructors,
and
$\caR_\caP$ is a set of rewrite rules representing the protocol's state changes.
More specifically, given a protocol $\mathcal{P}$, 
a \emph{state} in the protocol execution 
is
an $E_\caP$-equivalence class $[t]_{E_\caP}$
with $t$ a term
 of sort \sort{State},
$[t]_{E_\caP} \in T_{\Sigma_{\mathcal P}/ E_{\mathcal P}}(\cal{X})_{\sort{State}}$.
In basic Maude-NPA there are two types of algebraic properties: 
(i) {\em equational axioms}, such as commutativity,  associativity-commutativity,  
or associativity-commutativity-identity,
      called \emph{axioms},  and
  (ii) {\em equational rules}, called \emph{equations}.
%
Basic Maude-NPA includes two predefined sorts:
(i) the sort \sort{Msg} 
that allows the protocol specifier to describe 
other sorts 
as subsorts of the 
sort \sort{Msg},
and
(ii) the sort $\sort{Fresh}$
for representing fresh unguessable values,
e.g., nonces.

\begin{example}\label{ex:specNSL}
The specification of the NSL protocol in Maude-NPA is as follows.
A nonce generated by principal $A$
is denoted by $n(A,r)$, 
where $r$ is a unique variable of 
sort $\sort{Fresh}$
and $A$ denotes who generated the nonce.  This representation makes it easier to specify and keep track of the origin of nonces.  E.g.,  one can use the notation to specify a state in which a principal accepts a nonce as coming from $A$ when it actually comes from some $B \ne A$.  Concatenation of two messages, e.g., $N_A$ and $N_B$, is denoted by the operator
$\_{;}\_$, e.g., $n(A,r)\ ;\ n(B,r')$.
Encryption of a message $M$ with the public key  of principal $A$
is denoted by $pk(A,M)$, e.g., 
$\{N_B\}_{pub(B)}$ is denoted by
$pk(B,n(B,r'))$.
Encryption with the secret key  of principal $A$ 
is denoted by $sk(A,M)$.
%
%
The signature 
$\Symbols_{NSL}$ 
of the NSL protocol 
contains only terms such as $n(A,r)$, $M_1 ; M_2$,  $pk(A,M)$, and $sk(A,M)$.
 
The equational theory of the NSL protocol  contains no axioms
and only the equations describing
 public/private encryption cancellation:
 $E_{NSL} = \{ \, pk(A,sk(A,M)) = M, \ sk(A,pk(A,M)) =  M \}$.
\end{example}

A protocol $\mathcal{P}$ is specified with a notation 
derived from strand spaces \citep{strands}.
In a \emph{strand},
a local
execution of a protocol by a principal is indicated by 
a sequence of messages
$ 
[msg_1^-,\ msg_2^+,\ msg_3^-,   \allowbreak \ldots,\ msg_{k-1}^-,\ msg_k^+]
$ 
where each $msg_{i}$ is a term of
sort \textsf{Msg} (i.e., $msg_{i}\in T_{\Sigma_{\mathcal{P}}}(\cal{X})_{\textsf{Msg}}$).
Strand items representing input messages are assigned a negative sign,
and strand items representing output messages are assigned a positive sign.
We write ${m}^{\pm}$ to denote $m^{+}$ or $m^{-}$,
indistinctively.
We often write $+(m)$ and $-(m)$ instead of $m^{+}$ and $m^{-}$, respectively.
%
%
We make 
explicit
the $\sort{Fresh}$ variables 
$r_1,\ldots,r_k (k \geq 0)$
generated by a strand 
by writing 
${:: r_1,\ldots,r_k::}\ [msg_1^\pm,\ldots,msg_n^\pm]$,
where $r_1,\ldots,r_k$ 
appear somewhere in 
$msg_1^\pm,\ldots,msg_n^\pm$.
Fresh variables generated by a strand are unique and this is enforced during execution.
%
Furthermore, fresh variables are treated as constants that are never instantiated.

In Maude-NPA \citep{EscMeaMes-tcs06,FOSAD07}, 
strands evolve over time 
and thus we use the symbol $|$ to divide past and future in a strand, 
i.e.,
$
[nil, msg_1^\pm,   \allowbreak \ldots, msg_{j-1}^\pm ~|~   \allowbreak msg_j^\pm, msg_{j+1}^\pm, \ldots,   \allowbreak msg_k^\pm ,   \allowbreak nil ]
$,
where $msg_1^\pm,   \allowbreak\ldots,   \allowbreak msg_{j-1}^\pm$ are the past messages,
and $msg_{j}^\pm,   \allowbreak msg_{j+1}^\pm,   \allowbreak \ldots, msg_k^\pm$ are the future messages
($msg_{j}^\pm$ is the immediate future message).
In this presentation we will often remove the nils to simplify the exposition, except when there is nothing else between the vertical bar and the beginning or end of a strand.  If there is no risk of confusion, we may also remove the fresh variables appearing before the strand.

We write $\caP$ for the set of strands in a protocol,
including the strands that describe the intruder's behavior.
When it is necessary to identify a strand ${:: r_1,\ldots,r_k::}\ [msg_1^\pm,\ldots,msg_n^\pm]$ to distinguish it from other strands,
we will do so via a \emph{role name} in parentheses appearing before the strand, e.g. $(\mathit{initiator})~{:: r_1,\ldots,r_k::}\ [msg_1^\pm,\ldots,msg_n^\pm]$.

\begin{example}\label{ex:strandsNSL}
Let us continue  Example~\ref{ex:specNSL}.
The two principal strands associated  to the 
NSL protocol describing the three steps shown 
in~Section~\ref{sec:examples-NSLDB}
are as shown below.

\begin{align*}
   :: r :: &
    [nil ~|~ +(pk(B,n(A,r) ; A)), -(pk(A,n(A,r) ; N_B ; B)), 
           +(pk(B,N_B))] 
    \\
   :: r' :: &
    [nil ~|~ -(pk(B,N_A ; A)), +(pk(A,N_A ; n(B,r') ; B)),
           -(pk(B,n(B,r'))) ]
\end{align*}

\noindent
In the NSL protocol the intruder has the following capabilities:
(i) it can perform   encryption  with any public key,
(ii) it can only perform encryption  with its own private key,
(iii) it can concatenate two messages, and
(iv) it can decompose a concatenation into each of its parts.
For example, the intruder's ability to concatenate
two messages $M_1$ and $M_2$ is described
by the following strand:

\begin{align*}
   :: nil :: &
    [nil ~|~ -(M_1), -(M_2), +(M_1 ; M_2) ] 
\end{align*}

\end{example}

A \emph{state} in Maude-NPA is a pair consisting of a set of Maude-NPA strands and the \emph{intruder knowledge} at that time.
The set of Maude-NPA strands is  unioned\footnote{In reality we consider a multiset of strands but duplicates are discarded as redundant, see \citep{EscMeaMes-tcs06,FOSAD07}.} together 
by an associative and commutativity
union operator $\_\&\_$
with identity operator $\emptyset$,
along with an additional term describing the intruder knowledge at that point.
The \emph{intruder knowledge} is represented 
as a set of facts unioned\footnote{Again, in reality we consider a multiset of intruder facts but duplicates are discarded as redundant, see \citep{EscMeaMes-tcs06,FOSAD07}.} together
with an associative and commutativity
union operator \verb!_,_!
with identity operator $\emptyset$.
There are two kinds of intruder facts:
\emph{positive} knowledge facts 
(the intruder knows message $m$, i.e., $\inI{m}$), and 
\emph{negative} knowledge facts 
(the intruder does not yet know $m$ but will know it in a future state, 
denoted by $\nI{m}$). 
We represent a state as a term 
$$s_1 \& s_2 \& \cdots s_n \& (\inI{m_1},\ldots,\inI{m_k},\nI{m'_1},\ldots,\nI{m'_j})$$ 
with $s_1 \& s_2 \& \cdots s_n$ the \emph{set of strands} and $\inI{m_1},\ldots,\inI{m_k},\nI{m'_1},\ldots,\nI{m'_j}$ the \emph{intruder knowledge},
i.e., we consider the intruder knowledge as another state component, enclosed in parenthesis, to simplify the exposition.
 
We now describe the rewrite rules used in basic Maude-NPA to describe forward execution. When new strands are not added  into the state, 
the rewrite rules $R_\caP$ obtained from the protocol strands $\caP$
 are 
as follows, 
where 
$L,L'$ are 
variables of the sort for lists of input and output messages
($+m$,$-m$),
$IK$ is a variable of the sort for sets of intruder facts 
(\inI{m},\nI{m}),
$SS$ is a variable of the sort for sets of strands,
and
$M$ is a variable of sort \sort{Msg}:

\begin{small}
\begin{align}
&SS\ \&\ [L ~|~ M^-, L']\ \&\ (\inI{M},IK) 
  \to SS\ \&\ [L, M^- ~|~ L']\ \&\ (\inI{M},IK)
  \label{eq:negative-1}\\
&SS\ \&\ [L ~|~ M^+, L']\ \&\ IK 
  \to SS\ \&\ [L, M^+ ~|~ L']\ \&\ IK
  \label{eq:positiveNoLearn-2}\\
&SS\ \&\ [L ~|~ M^+, L']\ \&\ (\nI{M},IK) 
  \to SS\ \&\ [L, M^+ ~|~ L']\ \&\ (\inI{M},IK)
  \label{eq:positiveLearn-4}
\end{align}%
\end{small}%

In a \emph{forward execution} of the protocol strands,
Rule \eqref{eq:negative-1} synchronizes
an input message with a message already in the channel (i.e., learned by the intruder),
Rule \eqref{eq:positiveNoLearn-2} accepts output messages 
but the intruder's knowledge is not increased,
and
Rule \eqref{eq:positiveLearn-4} 
accepts output messages 
and the intruder's knowledge is positively increased.
\noindent
Note that
Rule \eqref{eq:positiveLearn-4} 
makes explicit \emph{when} the intruder learned
a message $M$, which is recorded in the previous state 
\footnote{Of course, in an actual forward execution of a protocol,
the intruder knowledge only has positive facts. The usefulness of
\nI{m} becomes clear when we consider \emph{backward} executions,
so that at the beginning of a protocol execution all intruder knowledge 
will be ``negative'', i.e., to be learned in the future.}
 by the negative fact \nI{M}.

New strands 
are added to the state by explicit introduction through dedicated rewrite rules
(one for each honest or intruder strand).
It is also the case that when we are performing a backwards search, only the strands that we are searching for are listed explicitly, and
extra strands necessary to reach an initial state are dynamically added. Thus, when we
want to introduce new strands into the explicit description of the state, we need to describe additional rules for doing that,
as follows:

\begin{small}
\begin{align}
\mbox{For each }[~ l_1,\ u^+,\ l_2~] \in \caP:
SS\ \,\&\, [~ l_1 ~|~ u^+, l_2~] \,\&\, (\nI{u},IK)
\to
SS\ \,\&\, (\inI{u},IK)
\label{eq:newstrand}%
\end{align}%
\end{small}%

\noindent
where 
$u$ denotes a message,
$l_1,l_2$ denote lists of input and output messages
($+m$,$-m$),
$IK$ denotes a variable of the sort for sets of intruder facts 
(\inI{m},\nI{m}),
and
$SS$ denotes a variable of the sort for sets of strands.

\begin{example}
The rewrite rule introducing a new intruder strand during backwards execution 
associated to the concatenation of two learned messages
is 
as follows:

\vspace{-2ex}
\begin{small}
\begin{align}
&SS\ \&\ [M_{1}^-, M_{2}^- ~|~ (M_{1} ; M_{2})^+ ]\ \&\ (\nI{(M_{1} ; M_{2})},IK) 
\to 
SS\ \&\ (\inI{(M_{1} ; M_{2})},IK)
  \nonumber
\end{align}%
\end{small}%
\end{example}

\noindent
In summary, 
for a protocol $\caP$,
the set of rewrite rules 
obtained from the protocol strands
that are
used
for backwards narrowing reachability analysis
\emph{modulo} the equational properties $E_{\caP}$
is 
$R_{\caP} = \{ \eqref{eq:negative-1},\eqref{eq:positiveNoLearn-2},\eqref{eq:positiveLearn-4} \}
\cup \{\eqref{eq:newstrand}\}$.

\label{sec:analysis}

An \emph{initial state}
is the final result of the backwards reachability process when an attack is found,
and is described as follows:
\begin{enumerate}
\item in an initial state, all strands have the bar at the beginning,
i.e., all strands are of the form
$ :: r_{1},\ldots,r_{j} :: [\  nil \mid {m_1}^{\pm},\ \ldots,\ {m_k}^{\pm}\ ]$;
\item in an initial state, all the intruder knowledge is negative,
i.e., all the items in the intruder knowledge are of the form 
$\nI{m}$ and therefore only to be known in the future.
\end{enumerate}
From an initial state no further backwards reachability steps are possible.

\emph{Attack states}
describe 
not just single concrete attacks, but \emph{attack patterns}
(or if you prefer \emph{attack scenarios}), 
which are specified symbolically as terms (with variables)
whose instances are
the final attack states we are looking for. 
Given an attack pattern, 
Maude-NPA
tries to either find an instance of the attack or prove
that no instance of such attack pattern is possible.  



\begin{example}
In order to prove that the NSL protocol fixes the bug found in the Needham-Schroeder Public Key protocol (NSPK), i.e., the intruder cannot learn the nonce generated by Bob,
we should specify the following attack state:
$$
:: r ::  [nil, -(pk(b,a ; N_A)), +(pk(a,N_A ; n(b,r) ; b)), -(pk(b,n(b,r))) ~|~ nil] \, \& \, (n(b,r) \inI)
$$

\noindent
from which an initial state cannot be reached and has a finite search space, proving it secure.
\end{example}

\section{Abstract Definition of Sequential Protocol Composition in Maude-NPA}\label{sec:mnpa-comp}

Sequential composition of two protocols describes a situation in which one protocol (the \emph{child})
can only execute after another protocol (the \emph{parent}) has completed its execution,
which  allows the child protocol to use information generated during the execution of the parent protocol. 
The underlying idea of such a situation is that the end of the parent's protocol execution
is \emph{synchronized} with the beginning of the child's protocol execution.
In this section we present a synchronization syntax and semantics which refines that of \citep{EMMS10}. 
In Section \ref{subsec:parameters} we first explain in detail the syntactic and semantics
features necessary to express the synchronization among both protocols.  
Then, in  Section~\ref{subsec:composition-formalization} we provide an abstract definition of sequential  composition of two or more
protocols in Maude-NPA.
Finally, in Section~\ref{sec:compex}  
 we define  a concrete execution model for the 
one-to-one and 
one-to-many protocol compositions 
by extending the basic  Maude-NPA execution model.
%
%
Throughout this paper, we will refer to the syntax and semantics explained in this section 
as \emph{abstract composition syntax and semantics}.

\subsection{Input/Output Parameters and Roles}\label{subsec:parameters}

 In this section we describe in more detail the   
 new 
features we need to make explicit
in each protocol to later define abstract sequential protocol compositions.  These features
are identical to those defined in \citep{EMMS10}.
Each strand in a protocol specification in the Maude-NPA
is now extended with {\em input} and {\em output parameters}.
Input parameters are 
a sequence of variables of different sorts
placed at the beginning of a strand.  
Output parameters are a sequence of terms placed at the end of a strand.  
The strand notation
we will now use is
%
$
[
\{\overrightarrow{I}\},
\overrightarrow{M},
\{\overrightarrow{O}\} 
]
$
\noindent where $\overrightarrow{I}$ is a list of input parameter variables, $\overrightarrow{M}$ is a list of positive and negative terms in the strand notation of the Maude-NPA, and $\overrightarrow{O}$ is a list of output terms.
Note that all the variables of $\overrightarrow{O}$ must appear 
in $\overrightarrow{M}$ or $\overrightarrow{I}$, i.e., no extra variables are allowed in sequential protocol composition outputs.
The input and output parameters describe the exact assumptions
about each principal. 
Note that we allow each honest or Dolev-Yao strand to be \emph{labeled}
(e.g. \textit{NSL.init} or \textit{NSL.resp}) to denote the ``role''
of that strand in the protocol,
in contrast to the standard Maude-NPA syntax for strands. 
These strand labels play an important role in our protocol composition method.

\begin{example}\label{ex:NSL-MNPA}\label{ex:NSL-new}
Following Examples~\ref{ex:specNSL} and \ref{ex:strandsNSL},
the protocol $\caP$ with two strands 
associated to the three protocol steps shown in~Section~\ref{sec:examples-NSLDB} 
is now described as follows:

\begin{small}
\begin{align}
 (\textit{NSL.init})  :: r :: [  \{A,B\}, &  +(pk(B,n(A,r);A)),   \notag\\
                                             &  -(pk(A,n(A,r);N;B)),   \notag\\  
                                             & +(pk(B,N)),  \notag\\  
                                \{A,B, \ & n(A,r),N\} ] .   \notag\\
 (\textit{NSL.resp}) :: r :: [  \{A,B\},  & -(pk(B,N;A)), \notag\\
 				       &   +(pk(A,N;n(B,r);B)),  \notag\\  
				       & -(pk(B,n(B,r))), \notag\\  
		              \{A, B,\ & N,n(B,r)\}] . \notag
\end{align}%
\end{small}%
\end{example}

\begin{example}\label{ex:DB-MNPA}\label{ex:DB-new}
Similarly to the NSL protocol,
there are several technical details missing in the previous informal description of the Distance Bounding  (DB) protocol.
The exclusive-or operator is $\oplus$
and its equational properties are described using 
associativity and commutativity of $\oplus$ plus
the 
equations\footnote{Note that the redundant equational property 
$X \oplus X \oplus Y = Y$ is necessary in Maude-NPA
for coherence purposes; see~\citep{Viry02,DuranM10a}.} 
${X \oplus 0} = X$,
$X \oplus X = 0$,  and 
$X \oplus X \oplus Y = Y$.
Since Maude-NPA does not yet include timestamps, we do not include all the actions relevant to calculating time intervals, sending timestamps, and checking them.
The protocol $\caP$ with two strands
associated to the two protocol steps shown in~Section~\ref{sec:examples-NSLDB}  
is described as follows:%


\begin{small}
\begin{align}
    (\textit{DB.init}) :: r ::  [ \{A,B,N_A\}, & +(n(B,r)), \notag\\
                          &   -(n(B,r) \oplus N_A),  \notag\\
			   \{A,B,N_A, \  &  n(B,r)\}] . \notag\\
   (\textit{DB.resp}) :: nil :: [  \{A,B,N_A\}, &   -(N_B), \notag\\
                                  & +(N_B \oplus N_A),  \notag\\
                                  \{A,B,N_A, \ & N_B\}] . \notag
\end{align}
\end{small}


\noindent
This protocol specification
makes clear that the nonce $N_{A}$
used by the initiator is a parameter and 
is
never generated by $A$ during the run of DB.
However, the initiator $B$ does generate a new nonce.
\end{example}

\begin{example}\label{ex:KD-MNPA}
The previous informal description of the Key Distribution (KD) protocol also lacks
several technical details, which we supply here.
Encryption of a message $M$ with key $K$ 
is denoted by $e(K,M)$, e.g., 
$\{N'_B\}_{h(N_{A},N_{B})}$ is denoted by
$e(h(n(A,r),n(B,r')), n(B,r''))$.
Cancellation properties of encryption and decryption
are described using 
the 
equations 
$e(X,d(X,Z)) = Z$ and 
$d(X,e(X,Z)) = Z$.
Session keys are written $skey(A,r)$, where $A$ is the principal's name
and $r$ is a \sort{Fresh} variable.
The protocol $\caP$ with two strands
associated to  the KD protocol 
 steps
 shown above 
is described as follows:%

\vspace{1ex}

\begin{small}
\begin{align}
    (\textit{KD.init})  :: r :: [   \{A,B,K\},  & +(e(K,skey(A,r)),  \notag\\
  				             &  -(e(K,skey(A,r) ; N)), \notag\\ &+(e(K, N)),    \notag\\
                                  \{A,B,K,\ & skey(A,r),N\}] . \notag\\
    (\textit{KD.resp}) :: r :: [    \{A,B,K\},  & -(e(K,SK)), \notag\\
                                                      &  +(e(K,SK ; n(B,r))), \notag\\ &-(e(K,n(B,r))),  \notag\\ 
                                   \{A,B,K, \ & SK,n(B,r)\} ] . \notag%
\end{align}%
\end{small}%
\end{example}

In the rest of this paper we remove irrelevant parameters (i.e. input parameters for strands with no
parents, and output parameters for strands with no children) in order to simplify the exposition.
Therefore, a  strand is now a term of one of  
the following forms:
 
 \begin{enumerate}
 	\item 
		$[ nil, \overrightarrow{M}, nil]$,  
		 i.e. a standard strand that cannot be connected to either a parent or a child strand,
		
	\item 
		$[ \{\overrightarrow{I}\}, \overrightarrow{M}, nil ]$,
		 i.e. a \emph{child strand} that can   be connected to a parent strand,
	\item 
		$[ nil,  \overrightarrow{M}, \{\overrightarrow{O}\} ]$,
		 i.e. a \emph{parent strand}  that can   be connected to a child strand,
	\item 
		$[ \{\overrightarrow{I}\}, \overrightarrow{M}, \{\overrightarrow{O}\} ]$,
		 i.e. a strand that can be connected to both a parent and a child strand, or
		 
	\item 
		$[ \{\overrightarrow{I}\},  \{\overrightarrow{O}\} ]$,
		 i.e. a strand that can be connected to both a parent and a child strand,
		 but without sending or receiving any message, called a  \emph{void strand}.	 
 \end{enumerate}


\subsection{Strand and Protocol  Composition}\label{subsec:composition-formalization}

In this section we formally define sequential protocol composition in Maude-NPA.
We first define the sequential
composition of two strands, since this will help us to define sequential protocol composition in general.
Intuitively, sequential composition of two strands describes a situation in which one strand 
(the \emph{child}), can only execute after another strand  (\emph{the parent}) has
completed its execution.  
Each composition of two strands is  obtained 
by \emph{matching the output parameters of the parent strand with the input parameters of the child strand} in a user-specified way.
Note that it may be possible for a single parent strand to have more than one child strand.

\begin{definition}[Sequential Strand Composition]\label{def:strandComposition}
Given 
two 
strands
$(a) ::{\overrightarrow{r_{a}}}:: [ \{\overrightarrow{I_a}\},\allowbreak \overrightarrow{M_a},\allowbreak \{\overrightarrow{O_a}\} ]$
and
$(b) ::{\overrightarrow{r_{b}}}:: [ \{\overrightarrow{I_b}\},\allowbreak \overrightarrow{M_b},\allowbreak \{\overrightarrow{O_b}\} ]$
that are properly renamed to avoid variable sharing, a sequential strand composition
is a triple of the form $(a, b, \textit{MODE})$,
where
$a$ and $b$ denote the parent and child roles, respectively, 
and $\textit{MODE}$ is either  \textrm{1-1}  or  \textrm{1-*},
indicating a one-to-one or one-to-many composition.
This triple satisfies the following conditions for consistency:

\begin{enumerate}
\item  both $\overrightarrow{O_a}$  and $\overrightarrow{I_b}$  have the same length, i.e.
	$\overrightarrow{O_a} = m_1, \ldots, m_n$ and $\overrightarrow{I_b} = m'_1, \ldots, m'_n$, and

\item there exists at least one substitution $\sigma$ such that $\overrightarrow{O_a} \congr{E_\caP} \overrightarrow{I_b}\sigma$.
 
\end{enumerate}
\end{definition}

We note that the definition of sequential strand composition given here differs from that given in \citep{EMMS10} in that
in Definition \ref{def:strandComposition} each input parameter in a child strand  is matched  with the corresponding output parameters in the parent strand, while in \citep{EMMS10} the user can choose which parameters are matched.  This gives
the user more flexibility, particularly in the case in which different children use different output parameters of the same parent.  But it comes at the cost of being more complex to specify and implement.   Moreover, the case of different children needing
different output parameters can be taken care of by using ``dummy'' input parameters to match parental output parameters the child does not need, or  more generally,  by means of the protocol adapters described in Section~\ref{sec:adapters}.

%
%

\begin{example}\label{ex:strandComp}
Let us consider again the NSL protocol of Example~\ref{ex:NSL-new}
and the DB protocol of Example~\ref{ex:DB-new}.
The  composition of the NSL initiator strand and the DB responder strand is specified
by the triple $(\textit{NSL.init}, DB.resp, $1{-}1$)$.
However, the NSL protocol had four output arguments while the DB protocol had three input arguments
and we are required to adapt the syntax of the NSL protocol to have only the three arguments required by the DB protocol:
 
\begin{small}
\begin{align}
%
%
  (\textit{NSL.init})  :: r :: [  &\{A,B\}, \notag\\ &  +(pk(B,n(A,r);A)),   
                                              -(pk(A,n(A,r);N;B)), 
                                              +(pk(B,N)),  \notag\\  
                            &\{A,B, n(A,r)\} ] .   \notag\\
                                %
  (\textit{DB.resp}) :: nil :: [  &\{A,B,N_A\}, \notag\\ &   -(N_B), 
  +(N_B \oplus N_A),  \notag\\
                            & \{A,B,N_A,  N_B\}] . \notag
\end{align}
\end{small}
\end{example}

\begin{example}\label{ex:strandComp-KD}
Let us consider again the NSL protocol of Example~\ref{ex:NSL-new}
and the KD protocol of Example~\ref{ex:KD-MNPA}.
The  composition of the NSL responder strand and the KD initiator strand is specified
by the triple 
  $(\textit{NSL.resp}, \textit{KD.init}, $1{-}*$)$.
  But again, the NSL protocol had different output arguments 
  than the input arguments of the KD protocol 
and we are required to adapt the syntax of the NSL protocol  as follows:

 \begin{small}
\begin{align}
  %
(\textit{NSL.resp}) :: r :: [  & \{A,B\},    \notag\\ 
                                            & -(pk(B,N;A)),  +(pk(A,N;n(B,r);B)),  -(pk(B,n(B,r))), \notag\\  
		              &  \{B, A, h(N,n(B,r)) \}] . \notag\\
%
%
    (\textit{KD.init})  :: r' :: [   &\{B,A,K\},  \notag\\ & +(e(K,skey(B,r')),  
				             -(e(K,skey(B,r') ; N')), 
				             +(e(K, N')),    \notag\\
                                  &\{B,A,K, skey(B,r'),N'\}] . \notag
\end{align}
\end{small}

\noindent
such that the term $h(N,n(B,r))$ has the same sort as that of the input parameter $K$. 
\end{example}

 Intuitively, we can now define the sequential composition of two protocols as 
 a set of sequential strand compositions.
 
 \begin{definition}[Sequential Composition of Two Protocols]\label{def:2protocolComp}
Given two protocols $\caP_1$ and $\caP_2$ 
   that are properly renamed to avoid variable sharing, 
   a sequential composition of both
   protocols, written $\caP_1 \ ;_S  \caP_2$, is defined as a triple of the form
   $(\caP_1, S, \caP_2)$
   where
   $S$ denotes a 
   set of strand compositions 
   between a parent strand of $\caP_1$ and a child strand of $\caP_2$
   of the form described in Definition~\ref{def:strandComposition}.
   Note that the signature of such protocol composition is the union\footnote{Note that we allow shared items but require the user to solve any possible conflict. Operator and sort renaming is an option, as in the Maude module importation language, but we do not consider those details in this paper.}\label{footnote:clash-renaming} of the signature of both protocols, 
   i.e., $\Sigma_{\caP_1  ;_S  \caP_2} = \Sigma_{\caP_1} \cup \Sigma_{\caP_2}$. 
   Similarly, the set of equations specifying the  algebraic properties of such protocol composition
   is the union\footnote{We assume the combined equational theory satisfies all the requirements for having a finitary and complete unification algorithm.} of the equations of both protocols,
   i.e., $E_{\caP_1  ;_S  \caP_2} = E_{\caP_1} \cup E_{\caP_2}$. 
 \end{definition}

  \begin{example}\label{ex:protComp-NSLDB}
Let us consider again both the NSL and DB protocols of Example~\ref{ex:strandComp} and their composition.
The  composition of both protocols, 
which is an example of a one-to-one composition,  is specified as follows, 
indicating that the initiator of NSL can be composed with the responder of DB and the responder of NSL with the initiator of DB:
%
\begin{align}
NSL \ ;_S \ DB =  ( \mathit{NSL},  \{ & (\mathit{NSL.init},  \mathit{DB.resp},  1{-}1), \notag\\
								     &  (\mathit{NSL.resp},  \mathit{DB.init}, 1{-}1) \}, \mathit{DB})   \notag
\end{align}
%
\noindent
The  strands are left as follows,
where we have removed irrelevant input and output parameters for clarity and simplicity:
\begin{small}
\begin{align}
 (\textit{NSL.init})  :: r :: [ &  +(pk(B,n(A,r);A)),   
  -(pk(A,n(A,r);N;B)), +(pk(B,N)), \notag\\ &\{A,B,n(A,r)\}]   \notag\\
 (\textit{NSL.resp}) :: r :: [ & -(pk(B,N;A)), 
 +(pk(A,N;n(B,r);B)), -(pk(B,n(B,r))), \notag\\ & \{A,B,N\}] \notag\\
   (\textit{DB.init}) :: r :: \  [ &  \{A,B,N_A\}, \notag\\ &+(n(B,r)), -(n(B,r) \oplus N_A) ] \notag\\
  (\textit{DB.resp}) :: nil :: \  [ &  \{A,B,N_A\}, \notag\\ & -(N_B), +(N_B \oplus N_A) ] \notag
\end{align}
\end{small}
\end{example}
 
\begin{example}\label{ex:protComp-NSLKD}
Let us now consider the NSL and KD protocols of Example~\ref{ex:strandComp-KD} and their composition.
The composition of both protocols, which is an example of a one-to-many 
composition, is specified as follows, indicating that there are four possible compositions:
the initiator of NSL composed with either the initiator or the responder of KD,
and the responder of NSL composed with either the   initiator or the  responder of KD:

\begin{align}
NSL \ ;_S \ KD =  ( \textit{ NSL}, &  (\textit{NSL.init},  \textit{KD.init}, 1{-}*), \notag\\
						&  (\textit{NSL.init},  \textit{KD.resp}, 1{-}*), \notag\\  						    
						&  (\textit{NSL.resp},  \textit{KD.init}, 1{-}*), \notag\\
						&  (\textit{NSL.resp},  \textit{KD.resp}, 1{-}*) \}, \textit{KD})   \notag
\end{align}

The  strands are as follows,
where we have removed irrelevant input and output parameters for clarity and simplicity:

\begin{small}
\begin{align}
 (\textit{NSL.init})  :: r :: [ &  
  +(pk(B,n(A,r);A)), -(pk(A,n(A,r);N;B)),  
  +(pk(B,N)), \notag\\
  &\{A,B,h(n(A,r),N)\}]   \notag\\
 (\textit{NSL.resp}) :: r :: [ &    -(pk(B,N;A)), +(pk(A,N;n(B,r);B)), 
   -(pk(B,n(B,r))), \notag\\ &\{B,A , h(N,n(B,r))\}]  \notag\\
    (\textit{KD.init})  :: r :: [ & \{C,D,K\},  \notag\\ &+(e(K,skey(C,r)),  
  -(e(K,skey(C,r) ; ND)), +(e(K, ND))] \notag\\
    (\textit{KD.resp}) :: r :: [ &  \{C,D,K\}, \notag\\ & -(e(K,SKD)), +(e(K,SKD ; n(C,r)), 
    -(e(K,n(C,r))]  \notag
\end{align}
\end{small}


Note that in the KD strands we use variables $C$ and $D$ to avoid confusion,
since depending on how the NSL and KD protocols are composed, they will be instantiated
as either the NSL initiator or the NSL responder name, represented by variables $A$ and $B$, respectively.
\end{example}

 In addition, we need to define the sequential composition of more than two protocols.
 Intuitively, the sequential composition of $n$ protocols
 $\caP_1, \ldots, \caP_n$
  is a sequence of
 two-protocol compositions, such that each protocol is composed with the previous protocol (except $\caP_1$)
 and with the next protocol (except  $\caP_n$).

 \begin{definition}[Sequential Composition of $n$ Protocols]\label{def:NprotocolComp}
 Given $n$ protocols $\caP_1, \ldots, \caP_n$ 
 that are properly renamed to avoid variable sharing,
 the sequential composition of all of them is denoted by:
 $$\caP_1\ ;_{S_1} \caP_2 \ ;_{S_2} \caP_3 \ ;_{S_3} \ldots ;_{S_{n-2}} \caP_{n-1} \ ;_{S_{n-1}} \caP_n$$
iff 
  $\caP_1\ ;_{S_1} \caP_2$,
  $\caP_2 \ ;_{S_2} \caP_3$, \ldots, 
  $ \caP_{n-1} \ ;_{S_{n-1}} \caP_n$ are sequential protocol compositions as explained in Definition~\ref{def:2protocolComp}.
 \end{definition}
 
 \subsection{Protocol Adapters}\label{sec:adapters}
 
 As we see from the examples in Section \ref{subsec:composition-formalization}, putting the composition information
 inside the role specification itself instead of specifying them separately introduces a potential modularity issue if we want to
 reuse roles in different specifications, in that different compositions may require different  information.  For example,  in one composition a child may require less information than a child in another composition with the same parent, as is the case in
 with NSL-DB versus NSL-KD.  Or, it may be more convenient to present the information in different orders in either the parent or the child, as is the case for NSL-DB versus NSL-KD. Or, one child may need the result of applying a function to parent output, while the other may require the output without that function applied.     Although some of these issues may be avoidable by careful planning, forcing the user to consider them in advance works against the sort of modularity we are trying to achieve.
 
 As a solution to this problem we propose the use of \emph{protocol adapters}, somewhat similar to the plug adapters one uses for overseas travel.  A protocol adapter, applied to the output of a parent protocol, would perform the operations on it
that would result in suitable input for the child protocol.  Such operations would include, but would not necessarily be limited to:

\begin{enumerate}
\item restricting the output parameters to a subsequence used by a child;
\item permuting the output parameters in the order used by a child, and;
\item computing symbolic functions on the output.
\end{enumerate} 

In a similar way, the input parameters of a child protocol can be restricted to a subsequence
or permuted to fit the output parameters of a parent.
We note that it is currently possible to specify such role adapters directly from void strands, using a void strand that takes
as its input the output parameters of the parent, and produces as its output the result of transforming these parameters into a format acceptable by the child.  However, this is a suboptimal solution in that it introduces an extra narrowing step to address a purely syntactic issue.  Thus, we are currently considering the best way of implementing protocol adapters on the syntactic level.


\subsection{Operational Semantics} \label{sec:compex}
 
As explained in Section~\ref{sec:maude-npa}, 
the operational semantics of protocol execution and analysis
is based on rewrite rules denoting state transitions
which are applied \emph{modulo} the algebraic properties $E_{\caP}$
of the given protocol $\caP$.
Therefore, in the one-to-one and one-to-many cases
we must add new state transition rules in order to deal with protocol composition. 
Maude-NPA performs backwards search 
modulo $E_{\caP}$
by reversing the transition rules expressed in a forward way; see Section~\ref{sec:maude-npa}.

\begin{figure}[t]
\begin{align}
&\mbox{For each one-to-one 
strand composition }
(a, b,   \mathrm{1{-}1})
\mbox{ with }
\notag\\
& \mbox{strand   }
(a) [
\overrightarrow{M_a},\{\overrightarrow{O_{a}}\}] 
\mbox{ for protocol } \Pone,
\mbox{strand   }
(b)[\{\overrightarrow{I_{b}}\},\overrightarrow{M_b} 
]
\mbox{ for protocol } \Ptwo,
\notag\\
& \mbox{and for each substitution } \sigma \mbox{ s.t. } \overrightarrow{I_{b}}\sigma \congr{E_{\caP}} \overrightarrow{O_{a}} ,
\mbox{we add the following rules:}
\notag\\[1.5ex]
&\hspace{4ex}
SS\,\&\, 
(a)\ [\overrightarrow{M_a} ~|~ \{ \overrightarrow{O_{a}} \}] 
~ \& ~ 
(b)\  [nil ~ | ~ \{\overrightarrow{I_{b}}\sigma\}, \overrightarrow{M_b}\sigma ] 
\hspace{0ex}\,\&\, IK%
\notag\\ 
&
\rightarrow
SS\,\&\, 
(a)\  [\overrightarrow{M_a}, \{ \overrightarrow{O_{a}} \} ~|~ nil ] 
~ \& ~ 
(b)\  [\{\overrightarrow{I_{b}}\sigma\}  ~|~  \overrightarrow{M_b}\sigma 
] 
\,\&\, IK%
\label{eq:one-to-one-forward-transf}\\[1.5ex]
&\hspace{4ex}
SS\,\&\, 
(a)\  [\overrightarrow{M_a} ~|~ \{ \overrightarrow{O_a} \}] 
~ \& ~ 
(b)\  [nil ~ | ~ \{ \overrightarrow{I_{b}}\sigma\}  ,   \overrightarrow{M_b}\sigma 
] 
\hspace{0ex}\,\&\, IK%
\notag\\
&
\rightarrow
SS\,\&\,  
(b)\  [\{ \overrightarrow{I_{b}}\sigma \}  ~|~   \overrightarrow{M_b}\sigma
] 
\,\&\, IK%
\label{eq:one-to-one-forward-new-transf}
\end{align}
\vspace{-6ex}
 \caption{Forward semantics for one-to-one composition}
\label{fig:one-to-one}
\end{figure}

In the one-to-one composition, 
we add the state transition rules of Figure~\ref{fig:one-to-one}
to the rewrite theory $(\Symbols_\caP, E_\caP, R_{B\caP})$ of Section~\ref{sec:maude-npa}.
Note that these transition rules are written in a forwards way but will be executed backwards, as the basic 
transition rules of Section~\ref{sec:maude-npa}.
Rule~\ref{eq:one-to-one-forward-transf} composes a parent and a child strand
already present in the current state.
Rule~\ref{eq:one-to-one-forward-new-transf} 
is the same as Rule~\ref{eq:one-to-one-forward-transf} but
adds, in a backwards execution, a parent strand to the current state and composes it with an existing child strand.
For example, 
given the   composition of the NSL initiator's strand  
with the DB responder's strand  
$(\textit{NSL.init}, DB.resp, \allowbreak 1{-}1)$ 
where NSL.init and DB.resp were defined in 
Example~\ref{ex:protComp-NSLDB}, 
we add the following transition rule
for Rule~\eqref{eq:one-to-one-forward-transf} where both the parent and the child strands are present
and thus synchronized.


  \begin{small}
\begin{align}
&(\textit{NSL.init}) :: r ::  \hfill \notag\\
& [ \, +(pk(B,n(A,r);A)), -(pk(A,n(A,r);N;B)), +(pk(B,N))  \mid \{A,B,n(A,r)\} \, ]\ \& \notag\\
& (\textit{DB.resp}) :: nil ::   \notag\\
& [ \, nil ~|~ \{A,B,n(A,r)\}, -(NB), +(NB * n(A,r)) \, ] & \&\  SS \ \&\ IK \notag\\
& \longrightarrow \notag\\
& (\textit{NSL.init}) :: r :: \notag\\
& [ \, +(pk(B,n(A,r);A)), -(pk(A,n(A,r);N;B)), +(pk(B,N)),  \{A,B,n(A,r)\} \mid   nil \, ]\ \&\notag\\
& (\textit{DB.resp}) :: nil ::  \notag\\
& [ \, \{A,B,n(A,r)\} ~|~ -(NB), +(NB * n(A,r))\, ] & \&\ SS\ \&\ IK \notag
\end{align}
\end{small}

\begin{figure}[t]
\begin{align}
&\mbox{For each one-to-many strand composition }
(a, b,  1{-}*)
\mbox{ with }
\notag\\
& \mbox{strand   }
(a) [
\overrightarrow{M_a},\{\overrightarrow{O_{a}}\}] 
\mbox{ for protocol } \Pone,
\mbox{strand   }
(b)[\{\overrightarrow{I_{b}}\},\overrightarrow{M_b} 
]
\mbox{ for protocol } \Ptwo,\notag\\
& \mbox{and for each substitution } \sigma \mbox{ s.t. } \overrightarrow{I_{b}}\sigma \congr{E_{\caP}} \overrightarrow{O_{a}} ,
\mbox{we add one Rule \ref{eq:one-to-one-forward-transf},
one Rule \ref{eq:one-to-one-forward-new-transf},
and 
rule}:
\notag\\[1ex]
&\hspace{3ex}
SS\,\&\,  
(a)\  [  \overrightarrow{M_a} ~|~ \{ \overrightarrow{O_{a}} \}]  
~ \& ~ 
(b)\  [nil ~ | ~ \{\overrightarrow{I_{b}}\sigma\}, \overrightarrow{M_b}\sigma
] 
\hspace{0ex}\,\&\, IK%
\notag\\
& 
\rightarrow
SS\,\&\, 
(a)\  [  \overrightarrow{M_a} ~|~ \{ \overrightarrow{O_{a}} \}  ]  
~ \& ~ 
(b)\  [\{ \overrightarrow{I_{b}}\sigma \}  ~|~ \overrightarrow{M_b}\sigma ] 
\,\&\, IK%
\label{eq:one-to-many-forward-transf}
\end{align}
\vspace{-6ex}
 \caption{Forward semantics for one-to-many composition}
\label{fig:one-to-many}
\end{figure}

One-to-many composition uses the rules in Figure~\ref{fig:one-to-one}
for the first child, plus an additional rule for subsequent children,
described in Figure~\ref{fig:one-to-many}.
Rule~\ref{eq:one-to-many-forward-transf}
composes a parent strand and a child strand
but the bar in the parent strand is not moved, in order to allow
further backwards child compositions.
%
For example, 
given the   composition of the NSL responder's strand 
with the KD initiator's strand 
$(\textit{NSL.resp}, KD.init, \textrm{1{-}*})$ 
\noindent
where $NSL.resp$ and $KD.init$  are as defined in 
Example~\ref{ex:protComp-NSLKD},  
we add the following transition rule
for Rule~\eqref{eq:one-to-many-forward-transf}:


 \begin{small}
\begin{align}
& (\textit{NSL.resp}) :: r ::  \notag\\
& [ \, -(pk(B,NA;A)), +(pk(A,NA;n(B,r);B)), -(pk(B,n(B,r))) ~|~  \{B,A, h(NA,n(B,r))\} \,], \notag\\
& (\textit{KD.init}) :: r' ::  \notag\\
&  [ \,nil ~|~ \{B,A,h(NA,n(B,r))\}, +(e(h(NA,n(B,r)),skey(B,r'))),  \notag\\
 & -(e(h(NA,n(B,r)),skey(B,r');N)), +(e(h(NA,n(B,r)), N)) \,] \hspace{12ex}\& SS\ \&\ IK \notag\\
& \longrightarrow \notag\\
& (\textit{NSL.resp}) :: r ::  \notag\\
& [ \, -(pk(B,NA;A)), +(pk(A,NA;n(B,r);B)), -(pk(B,n(B,r))) ~| ~  \{B,A,h(NA,n(B,r))\} \,], \notag\\
& (\textit{KD.init}) :: r' ::  \notag\\
& [ \, \{B,A,h(NA,n(B,r))\} ~|~ +(e(h(NA,n(B,r)),skey(B,r'))), \notag\\
 &  -(e(h(NA,n(B,r)),skey(B,r');N)), +(e(h(NA,n(B,r)), N)) \, ] \hspace{12ex}\& SS\ \&\ IK \notag
\end{align}
\end{small}

Thus, 
for a protocol composition $\caP_{1} ;_S \caP_{2}$,
the  rewrite rules governing protocol execution are
$R_{\caP_{1} ;_S \caP_{2}} = \{  \eqref{eq:negative-1},\eqref{eq:positiveNoLearn-2},
 \eqref{eq:positiveLearn-4}  \} \allowbreak
\cup  \allowbreak \eqref{eq:newstrand}  
\allowbreak \cup \allowbreak  
\eqref{eq:one-to-one-forward-transf} \cup \allowbreak  \eqref{eq:one-to-one-forward-new-transf}  \allowbreak
\cup \allowbreak  \eqref{eq:one-to-many-forward-transf}
$.
Note that the only generic rules are Rules  \eqref{eq:negative-1},\eqref{eq:positiveNoLearn-2},
 \eqref{eq:positiveLearn-4} and all the other are obtained from the protocol specification,
 thus increasing the number of transition rules.



 \section{Composition via synchronization messages}\label{sec:directComposition}


In Section~\ref{sec:mnpa-comp}
 we have provided an abstract syntax and a semantics for protocol composition,
 but this is not  what has been implemented in the tool.
There are two reasons for this, having to do with the fact that the rules in  Figures~\ref{fig:one-to-one} and \ref{fig:one-to-many}
are parametrized by the strands in the two composed protocols.  First of all, this means that implementing the rules would require a significant modification of Maude-NPA to support the new composition data type.  Secondly, the fact that each strand composition produces a new rule means that
the number of rewrite rules is significantly increased.  Increasing the number of rewrite rules can affect efficiency, since each rewrite rule must be tried at each narrowing step.   
Therefore, our approach has been to  instead implement composition using communication between strands, which can be achieved using only slight modifications
of constructs already present in Maude-NPA. 

In  \citep{EMMS10}  this communication was implemented via
messages sent over the Dolev-Yao channel; this implementation, referred to as \emph{synchronization by protocol transformation} has also been proved sound and
complete in \citep{tesis-sonia} with respect to the semantics given in Section~\ref{sec:mnpa-comp}.
However, as we will show in Section~\ref{sec:experiments-composition}     
this had a serious impact on performance due to the interleaving of the additional Dolev-Yao messages, as well as making it more difficult to write specifications
and attack states. Here, we present a modified version of Maude-NPA in which composition is achieved via \emph{synchronization messages}
that are passed \emph{directly} between a parent and child strand without going through the Dolev-Yao channel.  
Although, as in the case of composition with respect to protocol transformation, it is
necessary to add new rewrite rules, the rules are very similar to those of the basic Maude-NPA semantics, and
require the addition of fewer parametrized rules than for protocol transformation.
Composition of synchronization messages is still somewhat less expressive than the abstract semantics, in that the
same role cannot engage in both one-to-one and one-to-many compositions.  However, it can be proved sound and
complete with respect to the abstract semantics with the same restrictions. We discuss how this apparent restriction can be mitigated in Section \ref{subsec:synchro}. 
 
%
%
%

In Section~\ref{subsec:synchro} we 
introduce the notion of \emph{synchronization} of protocol strands,
a key idea underlying sequential protocol composition.
 In Section~\ref{subsec:compSyntax} we explain in detail the new Maude-NPA syntax
for the specification of protocol composition via
 synchronization messages. 
Section~\ref{sec:semantics-directComp} provides 
detailed information about  
the operational semantics of this
 direct implementation of protocol composition in Maude-NPA. 
%
Throughout this paper we will refer to these syntax and semantics as 
\emph{ synchronization via synchronization messages}. 
Finally, Section~\ref{sec:composition-SoundnessCompleteness}
proves the soundness and completeness of the semantics
in Section~\ref{sec:semantics-directComp}
with respect to the abstract semantics in Section~\ref{sec:compex}, 
thus proving that
the semantics in Section~\ref{sec:semantics-directComp}
is a \emph{correct} implementation of protocol composition in Maude-NPA.
We use our two running examples (NSL-DB and NSL-KD) to illustrate our technique.

\subsection{Synchronization Data Type Extension}\label{subsec:synchro}
  
  As explained above,   the underlying idea of a sequential protocol composition
   is that the end of the parent's protocol execution
is \emph{synchronized} with the beginning of the child's protocol execution.
Since in Maude-NPA a protocol execution is
denoted by a set of strands,   we actually need   to provide an
infrastructure to express the  notion of synchronization among strands,
so that the strands of the parent protocol can   in fact  be ``connected''
with the strands of the child protocol.

Synchronization of strands can be achieved in Maude-NPA 
by extending its syntax to define a special type of message
that we call \emph{synchronization message}.
The signature necessary to specify synchronization messages,
written $\Symbols_{Synch}$, is as follows.
Several sorts are added:
\sort{Synch} for the synchronization message,
\sort{Role} for user-definable constants denoting the roles in the protocol,
\sort{RoleConnection} for establishing which roles are the parent and which roles are the children, and
\sort{Mode} for choosing between one-to-one composition, denoted by constant 1-1,
and one-to-many composition, denoted by 1-*.
The synchronization messages are defined by  
patterns of the form:
\begin{align}
\{ a \rightarrow (b_1 b_2 \cdots b_j) \ ;;\  \textit{Mode} \ ;;\  \textit{Msg} \} ~\mbox{and}~ 
\{ (a_1 a_2 \cdots a_i) \rightarrow b \ ;;\  \textit{Mode} \ ;;\  \textit{Msg} \}.   \notag
\end{align}

\noindent 
%
%
%
The sort \sort{Role} contains some constants defined by the user for role names $a_1 a_2 \cdots$, 
e.g. NSL.init or NSL.resp.
The sort \sort{RoleConnection} 
contains just one operator 
$\_{\rightarrow}\_$,  
so that  $a \rightarrow (b_1 b_2 \cdots b_j)$ 
specifies that a parent role $a$ can have
 child roles $b_1$ through $b_j$,
 while $(a_1 a_2 \cdots a_i) \rightarrow b$ specifies the parent roles $a_1 a_2 \cdots a_i$
 that a child $b$ may have.  
 Thus ``NSL.init NSL.resp $\rightarrow$ KD.resp''
indicates that either the initiator or the responder roles of the NSL protocol can be the parent of the responder role 
of the KD protocol.
The information passed from parent to child is given in the third parameter, which is just a term of sort \sort{Msg}, allowing the user to construct any message representing the information exchanged in the synchronization.


\subsection{Syntax for Protocol Composition via synchronization messages}\label{subsec:compSyntax}

 In this section we explain in detail how the Maude-NPA's syntax 
 has been extended with   synchronization messages (see Section~\ref{subsec:synchro})
 in order to support the 
 input and output parameters of Section~\ref{subsec:parameters} 
 and the 
 abstract definition of
 protocol composition provided in 
 Section~\ref{subsec:composition-formalization}.
 Synchronization messages are used to represent
protocol compositions directly in the strand specification of the parent and child strands
without any protocol transformation.
A mapping from the notation for protocol composition of 
Section~\ref{subsec:composition-formalization} into synchronization messages is described as follows.


\begin{definition}[Parent Strand Synchronization]
Given two protocols $\caP_1$ and $\caP_2$,
a 
set $S$ of strand compositions, 
a role $a$ of $\caP_1$ of the form
 $(a)\ [ 
 \overrightarrow{M}, \{o_1,\ldots,o_{n}\} ]$,
and all the strand compositions for $a$ in $S$, 
i.e., $(a, b_1,   \mathit{Mode}),\ldots,(a, b_k,   \mathit{Mode})$,
we define
$$
synch_S(a)=
\left\{\begin{array}{l}
(a)\ [ \overrightarrow{M}, \{ a \rightarrow b_1 \cdots b_k\ ;;\ \textrm{Mode}\ ;; \allowbreak \ (o_1 ;\cdots ; o_{n}) \} ]\\
\end{array}
\right\}
$$
\end{definition} 

\begin{definition}[Children Strand Synchronization]
Given two protocols $\caP_1$ and $\caP_2$,
a 
set $S$ of strand compositions, 
a role $b$ of $\caP_2$ of the form
 $(b)\ [ \{i'_1,\ldots,i'_{n}\},\allowbreak \overrightarrow{M'} 
 ]$,
and all the strand compositions for $b$ in $S$,
i.e., $(a_1, b,   \mathit{Mode}),\ldots,(a_k, b,   \mathit{Mode})$,
we define
$$
synch_S(b)=
\left\{\begin{array}{l}
(b)\ [ \{ a_1 \cdots a_k \rightarrow b\ ;;\ \textrm{Mode}\ ;;\  (i'_1 ;\cdots ; i'_{n})\},\allowbreak \overrightarrow{M'} ]
\end{array}
\right\}
$$
\end{definition} 


\begin{definition}[Protocol Synchronization]\label{def:protocol-synchronization}
Given two protocols $\caP_1$ and $\caP_2$ 
   that are properly renamed to avoid variable sharing, 
and a sequential protocol composition
$\caP_1 \ ;_S  \caP_2=(\caP_1, S, \caP_2)$
where $S$ denotes a 
set of strand compositions of the form
$(a, b, \textit{MODE})$,
the protocol synchronization, denoted $sync(\caP_1 \ ;_S  \caP_2)$
is a single protocol which:
\begin{enumerate}
\item has signature $\Sigma_{\caP_1} \cup \Sigma_{\caP_2} \cup \Sigma_{Synch}$,
where $\Sigma_{Synch}$ is the new signature described in Section~\ref{subsec:synchro},
\item the equational theory is $E_{\caP_1} \cup E_{\caP_2}$
\item the set of strands is $synch(S)$, which is, by definition, the set of strands of the form
$synch_S(r)$ for each role $r$ in $\caP_1$ and $\caP_2$, 
and 
\item all the protocol compositions of a role have the same mode (1-1 or 1-*),
i.e., 
given $a$ in $\Pone$ 
and $(a, b_1, \textit{MODE}_1),\ldots,(a, b_k, \textit{MODE}_k)$ in $S$,
then $\textit{MODE}_1 = \ldots = \textit{MODE}_k$;
similarly
given $b$ in $\Ptwo$ 
and $(a_1, b, \textit{MODE}_1),\ldots,(a_k, b, \textit{MODE}_k)$ in $S$,
then $\textit{MODE}_1 = \ldots = \textit{MODE}_k$.
\end{enumerate}
\end{definition}

  As we shall see in Section \ref{sec:composition-SoundnessCompleteness}, synchronization via synchronization messages implements our abstract composition
semantics of Section~\ref{sec:compex},  but in the next section we clarify the role connections of our framework.

  \subsection{Role Connections}

As explained above, there are two types of synchronization messages
\begin{align}
\{ a \rightarrow (b_1 b_2 \cdots b_j) \ ;;\  \textit{Mode} \ ;;\  \textit{Msg} \} ~\mbox{and}~ 
\{ (a_1 a_2 \cdots a_i) \rightarrow b \ ;;\  \textit{Mode} \ ;;\  \textit{Msg} \}.   \notag
\end{align}

\noindent
They correspond to two different  parent or child situations associated to the abstract semantics that 
are now represented
using synchronization messages.

First, in the abstract semantics there
is nothing preventing a single instantiation of a parent role from having two or more children belonging to different
roles, assuming both child roles are allowed by the specification.  
For example, in the NSL-KD composition we can have 
an instance of the NSL initiator strand being composed with both
an instance of the KD initiator strand and 
an instance of the KD responder strand,
since in Example~\ref{ex:protComp-NSLKD},
the NSL.init has output parameters $\{A,B,h(n(A,r),N)\}$
while both KD.init and KD.resp have input parameters $\{C,D,K\}$.
Indeed, since we have one-to-many compositions for both KD.init and KD.resp, 
we could have an instance of NSL.init being composed with many different KD.init and many different KD.resp. 
In this case, we write 
``NSL.init $\rightarrow$ KD.init KD.resp'' in the synchronization message of the parent strand NSL.init.

Second, 
in the abstract semantics we can have a child role that participates in multiple protocol compositions, though a single instantiation of a child role has only one parent.
Again, in the NSL-KD composition we can have 
an instance of the KD initiator strand that can be synchronized with
either an instance of the NSL initiator strand or the NSL responder strand,
since in Example~\ref{ex:protComp-NSLKD},
the NSL initiator has output parameters $\{A,B,h(n(A,r),N)\}$,
the NSL responder has output parameters $\{B,A,h(n(A,r),N)\}$,
and the KD initiator has input parameters $\{C,D,K\}$.
Note that, in contrast to one parent being composed with many child instances of the children roles, in this case an instance of a child role would be composed only with an instance of the parent role. 
In this case, we write 
``NSL.init NSL.resp $\rightarrow$ KD.init'' in the synchronization message of the child strand KD.init.

We would like to stress that the restriction 
in Definition~\ref{def:protocol-synchronization} about all the strand roles participating in composition using always the same mode
has been  the result of a conscious decision to trade off expressiveness against
ease and readability of the specification.  For example, we could have allowed roles to be used in one-to-one
and one-to-many compositions by attaching modes to the names to each possible child role of a parent (and vice versa), e.g. "$a \rightarrow (b, 1-1) ,~ (c,1-*)$ but decided that this complicates the specification too much.  We note also that it is possible
to simulate roles that compose with children (or parents) using different modes by using void strands.  For example, instead of having $a$ compose directly with $b$ and $c$ we could have $a$ compose in $1-1$ mode with two void roles $b_0$ and $c_0$.  Void role $b_0$ would then compose with child $b$ in 1-1 mode and void role $c_0$ would compose with child $c$ in  $1-*$ model.  The performance impact of the extra narrowing step introduced by the void role can be mitigated by the use of partial order reductions, as we do for other steps in which no messages are exchanged over the Dolev-Yao channel.

We note that it is not possible to  simulate in the synchronization message syntax the  case in which a single instantiation of a strand may have children in more than one
mode, although this  is possible in the abstract syntax.  We believe that this is a reasonable price to pay.  We may consider introducing this capability later, but if so it will be in a larger context in which we consider a much more expressive syntax and semantics that is given by the current abstract semantics.  
See Section~\ref{subsec:future} for a discussion.

    In the following we provide the specification of our two examples of protocol composition,
namely the NSL Distance Bounding protocol (NSL-DB)
 and the NSL Key Distribution protocol (NSL-KD), presented in Sections~\ref{sec:examples-NSLDB}
 and \ref{sec:examples-NSLKD}, respectively,
 using the new synchronization message representation described above.

\begin{example}\label{ex:example-NSLDB}
 We begin with our example of one-to-one protocol composition, 
 i.e., the NSL-DB protocol. As explained in Section~\ref{sec:examples-NSLDB},
 the initiator of the DB protocol is always the child of the responder of the NSL protocol.
The specification of the protocol strands
using this syntax is as follows
 where the symbol $\oplus$ denotes the exclusive-or operator:

\begin{align} 
 (\textit{NSL.init}) :: r :: 
[ nil \mid &+(pk(B, n(A,r) ; A)) , 
        -(pk(A, n(A,r) ; NB ; B )), 
        +(pk(B, NB)), \notag\\
        &\{\textit{NSL.init} \rightarrow \textit{DB.resp}\ ;;\ \textrm{1-1}\ ;;\ (A ; B ; n(A,r))\} ] \&\notag\\
 (\textit{NSL.resp}):: r ::  
[ nil \mid &-(pk(B,NA ; A)), 
        +(pk(A, NA ; n(B,r) ; B)), 
        -(pk(B,n(B,r))),\notag\\
        &\{\textit{NSL.resp} \to \textit{DB.init}\ ;;\ \textrm{1-1}\ ;;\ (A ; B ; NA)\} ] \&\notag\\
 (\textit{DB.init}) :: r' ::
[ nil \mid\  &\{\textit{NSL.resp} \to \textit{DB.init}\ ;;\ \textrm{1-1}\ ;;\  (A ; B ; NA)\},\notag\\
        &+(n(B,r')), 
        -(NA \oplus n(B,r'))] \&\notag\\
 (\textit{DB.resp}) :: nil :: 
[ nil \mid\   &\{\textit{NSL.init} \to \textit{DB.resp}\ ;;\ \textrm{1-1}\ ;; \ (A ; B ; NA) \},\notag\\
         &-(N), 
         +(NA \oplus N) ]  
\notag
\end{align}
%
  \end{example}

  \begin{example}\label{ex:example-NSLKD}
 Let us now continue with our example of a one-to-many protocol composition, 
 i.e., the NSL-KD protocol. As explained in Section~\ref{sec:examples-NSLKD},
 the initiator of the session key protocol can be the child of either
 the initiator or responder of the NSL protocol.
 The specification of the strands of the NSL-KD protocol using the   
syntax for protocol composition via synchronization messages is as follows:

\begin{align}
 (\textit{NSL.init}) :: r :: 
[ nil \mid &+(pk(B, n(A,r) ; A)) , 
        -(pk(A, n(A,r) ; NB ; B )), 
        +(pk(B, NB)), \notag\\
        &\{\textit{NSL.init} \to \textit{KD.init}\ \textit{KD.resp}\ ;;\ \textrm{1-*}\ ;;\  (A ;  B ; h(n(A,r) , NB)) \}\  ] \&\notag\\
 (\textit{NSL.resp}) :: r ::
[ nil \mid &-(pk(B,NA ; A)), 
        +(pk(A, NA ; n(B,r) ; B)), 
        -(pk(B,n(B,r))), \notag\\
        &\{\textit{NSL.resp} \to \textit{KD.init}\ \textit{KD.resp}\ ;;\ \textrm{1-*}\ ;;\ (B ; A ; h(NA , n(B,r)))\}\  ]  \&\notag\\
 (\textit{KD.init}) :: r' ::
[ nil \mid &\{\textit{NSL.init NSL.resp} \to \textit{KD.init}\ ;;\ \textrm{1-*}\ ;;\ (C ; D ; K) \}, \notag\\
        &+(e(K, skey(C, r'))) , 
        -(e(K, skey(C, r') ; N)), 
        +(e(K, N))] \&\notag\\
 (\textit{KD.resp}) :: r' ::
[ nil \mid &\{\textit{NSL.init NSL.resp} \to \textit{KD.resp}\ ;;\ \textrm{1-*}\ ;;\ (C ; D ; K) \}, \notag\\
        &-(e(K, SKD)) , 
        +(e(K, SKD ; n(C, r'))), 
        -(e(K, n(C,r')))] \&\notag
\end{align}
%
 \end{example}

   \subsection{Operational Semantics of  Composition via synchronization messages}\label{sec:semantics-directComp}

 In Section~\ref{sec:compex}
 we provided an operational semantics based on extra transition
rules generated for each
possible protocol composition and we differentiated between rules generated for one-to-one compositions
and rules generated for one-to-many compositions. In this section we propose
a simplified version of that operational semantics, which we call \emph{composition via  synchronization messages semantics},
 that reduces the number of transition rules so that 
 now we just have
 two \emph{generic transition rules} and a set of \emph{generated transition rules} for each strand
in the same spirit of Rule~\eqref{eq:negative-1} and Rules~\eqref{eq:newstrand}.  

The two generic transition rules for protocol composition via synchronization messages
are described in Figure~\ref{fig:generic-directComposition}.
Note that these transition rules are written in a \emph{forwards} way but will be executed backwards, as the basic
transition rules of Section~\ref{sec:maude-npa} and the 
abstract composition semantics 
of  Section~\ref{sec:compex}.
The first generic transition Rule \eqref{eq:one-to-one-forward} is applicable to
both  one-to-one compositions and   one-to-many compositions.
This rule achieves the synchronization between both strands by means of the synchronization message.
The second generic Rule \eqref{eq:one-to-many-forward}
is applicable only to   one-to-many compositions
and represents the synchronization of a parent and a child without disabling the synchronization message of the parent. 

\begin{figure}[t!]
 \begin{align} 
&
SS\,\&\, 
(\textrm{a}) [ L ~|~ \{\textrm{a}\rightarrow \textrm{b}\ \textrm{R}  \ ;; \ \textrm{Mode} \ ;; \textrm{M}\}] 
 ~ \& ~
(\textrm{b}) [ nil ~|~ \{\textrm{a} \ \textrm{R'} \rightarrow \textrm{b}\ ;; \ \textrm{Mode} \ ;; \textrm{M}\}, L'] 
\hspace{0ex}\,\&\, IK%
\notag\\
\rightarrow&
SS\,\&\, 
(\textrm{a}) [ L, \{\textrm{a} \rightarrow \textrm{b} \ \textrm{R} \ ;; \ \textrm{Mode} \ ;; \textrm{M}\} ~|~ nil] 
~ \& ~
(\textrm{b}) [ \{\textrm{a} \ \textrm{R'} \rightarrow \textrm{b} \ ;; \ \textrm{Mode} \ ;; \textrm{M}\} ~|~ L'] 
\hspace{0ex}\,\&\, IK%
\label{eq:one-to-one-forward}
\\[1.5ex]
&
SS\,\&\, 
(\textrm{a}) [ L ~|~ \{\textrm{a} \rightarrow \textrm{b}\ \textrm{R}  \ ;; \texttt{1-*} ;; \textrm{M}\}] 
 ~ \&~  
(\textrm{b}) [ nil ~|~ \{\textrm{a}\ \textrm{R'}  \rightarrow \textrm{b} \ ;; \texttt{1-*} ;; \textrm{M}\}, L'] 
\hspace{0ex}\,\&\, IK%
\notag\\
\rightarrow&
SS\,\&\, 
(\textrm{a}) [ L ~|~ \{\textrm{a} \rightarrow \textrm{b}\ \textrm{R}  \ ;; \texttt{1-*} ;; \textrm{M}\} ] 
 ~\&  ~
(\textrm{b}) [ \{\textrm{a}\ \textrm{R'}  \rightarrow \textrm{b} \ ;; \texttt{1-*} ;; \textrm{M}\} ~|~ L'] 
\hspace{0ex}\,\&\, IK  \label{eq:one-to-many-forward} \\
& \mbox{where:} \notag\\
& \hspace{3ex}   L,L'  \mbox{ are variables of the sort for lists of input and output messages (+\textit{m},-\textit{m}),}
\notag\\
& \hspace{3ex} \mbox{\textit{IK} is a variable of the sort for sets of intruder facts $(\inI{m},\nI{m})$,}
\notag\\
& \hspace{3ex}   \mbox{\textit{SS} is a variable of the sort for sets of strands,}
\notag\\
& \hspace{3ex}   \mbox{\textrm{M} is a variable of sort \sort{Msg},}
\notag\\
& \hspace{3ex}   \textrm{a},  \textrm{b}  \mbox{ are variables of sort \sort{Role}, }
\notag\\
& \hspace{3ex}   \textrm{R}, \textrm{R'}  \mbox{ are variables denoting sets of roles, and}
\notag\\
& \hspace{3ex}    \textrm{Mode} \mbox{ is a variable of sort \sort{Mode}}
\notag
\end{align} 

\vspace{-2ex}
  \caption{Generic forward transition rules for composition via synchronization messages}
\label{fig:generic-directComposition}
 
 \vspace{-2ex}
  \begin{align} 
& \mbox{For each strand definition }
[\overrightarrow{M_a},\{a \rightarrow b_1\cdots b_k \ ;; \ \textrm{mode} \ ;; \textrm{msg}\}],
\mbox{ and each } i\in\{1,\ldots,k\},
\notag\\
& \mbox{we add a rule of the form:}
\notag\\[1ex]
&
SS\,\&\, 
(\textrm{a}) [ \overrightarrow{M_a} ~|~ \{a \rightarrow b_1\cdots b_k \ ;; \ \textrm{mode} \ ;; \textrm{msg}\}] 
~ \& ~ 
\notag\\
& 
\hspace{5.5ex}
(b_i) [ nil ~|~ \{a\ R \rightarrow b_i \ ;; \ \textrm{mode} \ ;; \textrm{msg}\}, L] 
\hspace{0ex}\,\&\, IK%
\notag\\
\rightarrow&
SS\,\&\, 
(b_i) [ \{a\ R \rightarrow b_i \ ;; \ \textrm{mode} \ ;; \textrm{msg}\} ~|~ L] 
\hspace{0ex}\,\&\, IK%
\label{eq:one-to-*-generated}\\
& \mbox{where:} \notag\\
& \hspace{3ex}   L  \mbox{ is a variable of the sort for lists of input and output messages (+\textit{m},-\textit{m}),}
\notag\\
& \hspace{3ex} \mbox{\textit{IK} is a variable of the sort for sets of intruder facts $(\inI{m},\nI{m})$,}
\notag\\
& \hspace{3ex}   \mbox{\textit{SS} is a variable of the sort for sets of strands,}
\notag\\
& \hspace{3ex}   \mbox{\textrm{msg} is a specific expression of sort \sort{Msg},}
\notag\\
& \hspace{3ex}   a, b_1,\ldots,b_k  \mbox{ are specific constants of sort \sort{Role}, }
\notag\\
& \hspace{3ex}   \textrm{R}  \mbox{ is a variable denoting sets of roles, and}
\notag\\
& \hspace{3ex}  \textrm{mode} \mbox{ is a specific constant of   sort  \sort{Mode}}
\notag
\end{align} 
\vspace{-4ex}
 \caption{Generated forward transition rules for  composition via  synchronization messages}
\label{fig:generated-directComposition}
\vspace{-4ex}
 
\end{figure}

These two generic rules synchronize
an output parameter of an existing parent strand with an input message of an existing child strand.
Both strands must be present in the state.
The difference between a one-to-one and one-to-many composition is that the output parameter of the parent strand
is kept in the same position of the parent strand for further synchronizations with other children strands.

As it happens in the basic 
Maude-NPA operational semantics of Section~\ref{sec:maude-npa},
we generate extra transitions rules from strands, in this case for protocol composition, as 
shown in Figure~\ref{fig:generated-directComposition}.
Transition rules of the form~\eqref{eq:one-to-*-generated}, when executed backwards,  allow adding to the state
a new parent strand, whose output parameters will be synchronized with
the input parameters of an already existing child strand.
Note that 
the generated transitions rules \eqref{eq:one-to-*-generated}  
apply to \emph{both} of the one-to-one or one-to-many composition cases.  In each case, they describe
a parent synchronizing with its first child.

For example, 
given the
composition of the NSL initiator's strand and the DB responder's strand,
where both strands were defined in Example~\ref{ex:example-NSLDB},
for Alice's strand



\begin{small}
\begin{align}
:: r :: 
[ nil | &+(pk(B, n(A,r) ; A)), 
        -(pk(A, n(A,r) ; NB ; B )), 
        +(pk(B, NB)), \notag\\
        &\{\textit{NSL.init} \to \textit{DB.resp}\ ;;\ \textrm{1-1}\ ;;\ (A ; B ; n(A,r))\} ]\notag
\end{align}
\end{small}

\noindent
we add the following transition rule generated by
Rule~\eqref{eq:one-to-*-generated}


\begin{small}
\begin{align}
& :: r ::
[ nil,  +(pk(B, n(A,r) ; A)), 
       -(pk(A, n(A,r) ; NB ; B )), 
       +(pk(B, NB)), \hfill \notag\\
&  \hspace{9ex} \mid\ \{\textit{NSL.init} \to \textit{DB.resp}\ ;;\ \textrm{1-1}\ ;;\ (A ; B ; n(A,r))\} ] \, \&\notag\\
& :: RR :: 
 [ nil \mid\ \{\textit{NSL.init}\ R \to \textit{DB.resp} \ ;; \textrm{1-1} ;;  (A ; B ; n(A,r)) \},  L ] 
 \hspace{12ex} \  \&\ SS\ \&\ IK\notag\\
& \longrightarrow \notag\\
&:: RR :: 
[ nil, \{\textit{NSL.init}\ R\to \textit{DB.resp}\ ;;\ \textrm{1-1}\ ;;\  (A ; B ; n(A,r)) \} \mid L ] 
    \hspace{12ex}   \&\ SS\ \&\ IK
  \notag
\end{align}
\end{small}


Thus, 
for a protocol composition $\caP_{1} ;_S \caP_{2}$,
the  rewrite rules governing protocol execution in  composition via synchronization messages are
$R_{synch(\caP_{1} ;_S \caP_{2})} = \{ 
 \eqref{eq:negative-1},\allowbreak\eqref{eq:positiveNoLearn-2},\allowbreak\eqref{eq:positiveLearn-4} \} \cup \allowbreak \eqref{eq:newstrand}
\cup\{\eqref{eq:one-to-one-forward},\eqref{eq:one-to-many-forward}\}
\cup  \eqref{eq:one-to-*-generated} 
$.

\clearpage

Here, the reader can realize that this  synchronization semantics for protocol composition contains two generic transition rules, Rules~\eqref{eq:one-to-one-forward} and \eqref{eq:one-to-many-forward},
and one transition rule
for each protocol composition from Rule~\eqref{eq:one-to-*-generated}, whereas the protocol composition presented
in Section~\ref{sec:mnpa-comp} produces several transition rules for each protocol composition.
Indeed, this simpler semantics for protocol composition requires fewer rules distinguishing
one-to-one and  one-to-many compositions than the abstract semantics.

\subsection{Soundness and Completeness 
}\label{sec:composition-SoundnessCompleteness}

   In this section
we prove soundness and completeness of the operational semantics composition via synchronization messages
presented in Section~\ref{sec:semantics-directComp} 
   with respect to the abstract compositional operational semantics of Section~\ref{sec:compex} under the restriction that a each child role (respectively parent role) and compose with parent (respectively child roles) in at most one mode.

    First, we must relate protocol states using the 
protocol composition rewrite rules of Section~\ref{sec:compex}
and protocol states in the composition via  synchronization messages.
Throughout this section, when we can avoid confusion, 
a state $St$ is called \emph{valid according to a rewrite theory $\caR$}
if it is a valid term of sort \sort{State} with respect to the order-sorted signature of $\caR$.

\begin{definition}[Bijective function \textit{trans}]
\label{def:equivalenceDC}
Let $\caP_{1}$ and $\caP_{2}$ be two protocols
and $\SeqComp$ their 
composition.   
Let  
$\caR_{\SeqComp}$ be the rewrite theory 
associated in Section~\ref{sec:compex}
to the abstract protocol composition $\SeqComp$
and
$\RTdc$ be the rewrite theory associated 
in Section~\ref{sec:semantics-directComp}
to  composition via  synchronization messages.
We define
the function $\textit{trans}_S$ mapping states valid according to the rewrite theory
$\caR_{\SeqComp}$ to
 states
valid according  to 
the rewrite theory
$\RTdc$
as
specified in Figure~\ref{fig:state-transformationDC},
and its inverse function
$\textit{trans}^{-1}_S$
as
specified in Figure~\ref{fig:state-invtransformationDC}.
\end{definition}

\begin{figure}[t]
\begin{small}
$trans_S(St)=
\left\{
\begin{array}{l@{\ \ }l}
(b) [ \{ \mathit{a_1 \cdots a_k \to b \, ;; Mode \, ;;  \,  {I}_{b}^{\sharp} } \}, \  \overrightarrow{b_{1}} \mid\overrightarrow{b_{2}}
]\ \&\ St'
& \mbox{if }
\begin{array}[t]{l}
(b) [\{\overrightarrow{I_{b}}\}, \overrightarrow{b_{1}} \mid \overrightarrow{b_{2}}
] \in St, \\
(a_1,b, \mathit{Mode}), \ldots,(a_k,b, \mathit{Mode}) \in S,\\
trans_S(St-(b))=St'
\end{array}
\\[0ex]
(b) [ nil \mid \{ \mathit{a_1 \cdots a_k \to b \, ;; Mode \, ;;  \,  {I}_{b}^{\sharp} } \}, \  \overrightarrow{b_{1}} 
]\ \&\ St'
& \mbox{if }
\begin{array}[t]{l}
(b) [nil \mid \{\overrightarrow{I_{b}}\}, \overrightarrow{b_{1}} 
] \in St, \\
(a_1,b, \mathit{Mode}),\ldots,(a_k,b, \mathit{Mode}) \in S,\\
trans_S(St-(b))=St'
\end{array}
\\[0ex]
(a) [\overrightarrow{a_{1}} \mid \overrightarrow{a_{2}}
,  \{ \mathit{a \to b_1 \cdots b_k \, ;; Mode \, ;;  \,  {O}_{a}^{\sharp} } \} ]\ \&\ St'
& \mbox{if }
\begin{array}[t]{l}
(a) [\overrightarrow{a_{1}} \mid \overrightarrow{a_{2}}
, \{\overrightarrow{O_{a}}\} ] \in St, \\
(a,b_1, \mathit{Mode}),\ldots,(a,b_k, \mathit{Mode}) \in S,\\
trans_S(St-(a))=St'\\
\end{array}\\
(a) [ \overrightarrow{a_{1}} 
,  \{ \mathit{a \to b_1\cdots b_k \, ;; Mode \, ;;  \,  {O}_{a}^{\sharp} } \} \mid nil ]\ \&\ St'
& \mbox{if }
\begin{array}[t]{l}
(a) [ \overrightarrow{a_{1}} 
, \{\overrightarrow{O_{a}}\} \mid nil ] \in St, \\
(a,b_1, \mathit{Mode}),\ldots,(a,b_k, \mathit{Mode}) \in S,\\
trans_S(St-(a))=St'\\
\end{array}\\
St & \mbox{otherwise}
\end{array}
\right.
$
\\
where 
\begin{minipage}[t]{.9\linewidth}
$I^{\sharp}$ (resp. $O^{\sharp}$)
is equal to $\overrightarrow{I}$ (resp.   $\overrightarrow{O}$)
by replacing the comma ``,'' by  a semicolon ``;'' to denote concatenation
of input and output parameters, e. g. input parameters    
$\overrightarrow{I} = \{A \, , \,B \,  , \, NA\}$
is written as the sequence    $I^{\sharp} = A \, ; \, B \, ; \,  NA$.
\end{minipage}
\caption{Function $trans$ 
 between states 
valid according  to the rewrite theory
$\caR_{\SeqComp}$
and states 
valid according to the rewrite theory
$\RTdc$
}
\label{fig:state-transformationDC}
\end{small}
\end{figure}

\begin{figure}[t]
\begin{small}
$trans^{-1}_S(St)=
\left\{
\begin{array}{l@{\ \ }l}
(b) [\{\overrightarrow{I_{b}}\}, \overrightarrow{b_{1}} \mid \overrightarrow{b_{2}}
] \&\ St'
& \mbox{if }
\begin{array}[t]{l}
(b) [ \{ \mathit{a_1 \cdots a_k \to b \, ;; Mode \, ;;  \,  {I}_{b}^{\sharp} } \}, \  \overrightarrow{b_{1}} \mid\overrightarrow{b_{2}}
]\ \in St, \\
trans^{-1}_S(St-(b))=St'
\end{array}
\\[0ex]
(b) [nil \mid \{\overrightarrow{I_{b}}\}, \overrightarrow{b_{1}} 
]\ \&\ St'
& \mbox{if }
\begin{array}[t]{l}
(b) [ nil \mid \{ \mathit{a_1 \cdots a_k \to b \, ;; Mode \, ;;  \,  {I}_{b}^{\sharp} } \}, \  \overrightarrow{b_{1}} 
] \in St, \\
trans^{-1}_S(St-(b))=St'
\end{array}
\\[0ex]
(a) [\overrightarrow{a_{1}} \mid \overrightarrow{a_{2}}
, \{\overrightarrow{O_{a}}\} ]\ \&\ St'
& \mbox{if }
\begin{array}[t]{l}
(a) [\overrightarrow{a_{1}} \mid \overrightarrow{a_{2}}
,  \{ \mathit{a \to b_1 \cdots b_k \, ;; Mode \, ;;  \,  {O}_{a}^{\sharp} } \} ]
 \in St, \\
trans^{-1}_S(St-(a))=St'\\
\end{array}\\
(a) [ \overrightarrow{a_{1}} 
, \{\overrightarrow{O_{a}}\} \mid nil ]\ \&\ St'
& \mbox{if }
\begin{array}[t]{l}
(a) [ \overrightarrow{a_{1}} 
,  \{ \mathit{a \to b_1\cdots b_k \, ;; Mode \, ;;  \,  {O}_{a}^{\sharp} } \} \mid nil ] \in St, \\
trans^{-1}_S(St-(a))=St'\\
\end{array}\\
St & \mbox{otherwise}
\end{array}
\right.
$
\\
where 
\begin{minipage}[t]{.9\linewidth}
$I^{\sharp}$ (resp. $O^{\sharp}$)
is equal to $\overrightarrow{I}$ (resp.   $\overrightarrow{O}$)
by replacing the comma ``,'' by  a semicolon ``;'' to denote concatenation
of input and output parameters, e. g. input parameters    
$\overrightarrow{I} = \{A \, , \,B \,  , \, NA\}$
is written as the sequence    $I^{\sharp} = A \, ; \, B \, ; \,  NA$.
\end{minipage}

 \caption{Function $trans^{-1}$ 
 between states 
valid according  to the rewrite theory
$\RTdc$
and states 
valid according to the rewrite theory
$\caR_{\SeqComp}$
}
\label{fig:state-invtransformationDC}
\end{small}
\end{figure}

   The following auxiliary results  ensure that 
there is 
an appropriate connection 
between states of both rewrite theories.

\begin{lemma}
Let $\caP_{1}$ and $\caP_{2}$ be two protocols
and $\SeqComp$ their 
composition.
Let 
$\caR_{\SeqComp}$ be the rewrite theory 
associated in Section~\ref{sec:compex}
to the protocol composition $\SeqComp$
and
$\RTdc$ be the rewrite theory associated 
in Section~\ref{sec:semantics-directComp}
to 
the   composition via  synchronization messages.

Then $\textit{trans}_S$ defined in Definition~\ref{def:equivalenceDC}
is a bijective function from terms of sort \sort{State} in 
$\RTdc$
to terms of sort \sort{State} in $\caR_{\SeqComp}$,
and has $\textit{trans}^{-1}_S$ as its inverse function.
\end{lemma}

\begin{proof}
By structural induction on the 
functions $\textit{trans}_S$ and $\textit{trans}^{-1}_S$ given in
Figures~\ref{fig:state-transformationDC}
and \ref{fig:state-invtransformationDC}.
The base case is a state $St$ that has no strand with input or output parameters,
since $\textit{trans}_S(St)=St$.
For the inductive case we consider only the case when $St$ contains a strand
of the form $(b) [\{\overrightarrow{I_{b}}\}, \overrightarrow{b_{1}} \mid \overrightarrow{b_{2}} ]$
and all the other cases are similar.
Let $St = (b)[\{i_1,\ldots,i_n\}, m_1^\pm, \ldots,m_i^\pm \mid m_{i+1}^\pm, \ldots,m_k^\pm ] \& ss \& ik$
where $ss$ denotes a set of strand instances and $ik$ the intruder knowledge of the state.
Let $(a_1,b, \mathit{mode}), \ldots,(a_k,b, \mathit{mode})$ be all the composition triples in $S$ involving role $b$.
By induction hypothesis we have that $\textit{trans}^{-1}_S(\textit{trans}_S(ss \& ik)) = ss \& ik$.
Then, by applying function \textit{trans} to $St$ we have that the strand instance $b$ is transformed into
$(b) [ \{ \mathit{a_1 \cdots a_k \to b \, ;; mode \, ;;  \,  (i_1; \cdots ; i_n) } \}, \  m_1^\pm,\allowbreak \ldots,m_i^\pm \mid m_{i+1}^\pm, \allowbreak\ldots,m_k^\pm]$.
But then it is easy to see that when we apply $\textit{trans}^{-1}$ to this transformed strand, we simply remove the synchronization message and get the same strand instance $b$.
Therefore, $\textit{trans}^{-1}_S(\textit{trans}_S(St)) = St$.
\qed
\end{proof}

 Let us now relate  
 backwards narrowing steps
using the rewrite theory associated to the composition  via  synchronization messages of  Section~\ref{sec:semantics-directComp}
(i.e., $\RTdc$)
 w.r.t. 
backwards narrowing using the rewrite theory associated to the abstract  protocol composition 
of Section~\ref{sec:compex}
(i.e., $\caR_{\SeqComp}$).
 Note that in this case a backwards narrowing step performed with a
rule of $\RTcomp$ always corresponds to one backwards narrowing step
with a rule of $\RTcomp$, since no extra messages are introduced to
synchronize parent and child strands.

 \begin{lemma}[Bisimulation]\label{oneStepDC}
 Let $\caP_{1}$ and $\caP_{2}$ be two protocols
and $\SeqComp$ their composition.  
Let 
$\RTcomp$ be the rewrite theory associated 
in Section~\ref{sec:compex}
to the abstract  protocol composition $\SeqComp$, 
and
$\RTdc$ be the rewrite theory associated 
in Section~\ref{sec:semantics-directComp} to  composition  via  synchronization messages.

Given  two 
states $St_{1}$ and $St_{2}$
valid according to the rewrite theory $\RTcomp$
such that 
$\textit{trans}_S(St_{1})=St'_{1}$,
$\textit{trans}_S(St_{2})=St'_{2}$,
$St_{1} \narrow{1}{\rho,\caR^{-1}_{\SeqComp},E_{\SeqComp}} St_{2}$
iff 
$St'_{1} \narrow{1}{\sigma,\caR^{-1}_{synch(\SeqComp)},E_{\SeqComp}}  St'_{2}$. 
\end{lemma}

\begin{proof}
We prove the result by case analysis 
on the applicable rewrite rules. 
First, let us recall the different rules that are applicable:
for a term $St_1$ valid according to the rewrite theory $\RTcomp$
we can apply the reversed version of Rules $\eqref{eq:negative-1}$, $\eqref{eq:positiveNoLearn-2}$,
and
$ \eqref{eq:positiveLearn-4}$ 
plus the reversed version of rules in any of the sets 
$\eqref{eq:newstrand}$, $\eqref{eq:one-to-one-forward-transf}$, $\eqref{eq:one-to-one-forward-new-transf} $,
and $\eqref{eq:one-to-many-forward-transf}$,
whereas for the term $St'_1 = trans_S(St_1)$ 
valid according to the rewrite theory $\RTdc$
we can apply the reversed version of Rules
$\eqref{eq:negative-1}$, $\eqref{eq:positiveNoLearn-2}$, $\eqref{eq:positiveLearn-4}$,
$\eqref{eq:one-to-one-forward}$, and $\eqref{eq:one-to-many-forward}$
plus the reversed version of rules in any of the sets
$\eqref{eq:newstrand}$ and $\eqref{eq:one-to-*-generated}$.
Second, we consider four possibilities below but only show in detail cases (a) and (b), since cases (c) and (d) are similar to case (b).

\begin{itemize}
\item[(a)]
When the
Rules~\eqref{eq:negative-1},\eqref{eq:positiveNoLearn-2}, and \eqref{eq:positiveLearn-4},
as well as any rule in the set \eqref{eq:newstrand},
are applied,
they do not  involve any composition and, since the function $\textit{trans}_S$ is a bijection,
the same type of rule would be applicable to $St'_1$.

\item[(b)]
A rule in the set~\eqref{eq:one-to-one-forward-transf}
corresponds to an application of 
Rule~\eqref{eq:one-to-one-forward}
(synchronizing the input parameters of the child strand with the output parameters
of the parent strand). 
In this case, the reversed version of a rule of the following form 
in set~\eqref{eq:one-to-one-forward-transf}
has been applied to state $St_1$

\begin{align}
& SS\,\&\, 
(a)\ [\overrightarrow{M_a} ~|~ \{ \overrightarrow{O_{a}} \}] 
~ \& ~ 
(b)\  [nil ~ | ~ \{\overrightarrow{I_{b}}\sigma\}, \overrightarrow{M_b}\sigma ] 
\hspace{0ex}\,\&\, IK%
\notag\\ 
\rightarrow
&
SS\,\&\, 
(a)\  [\overrightarrow{M_a}, \{ \overrightarrow{O_{a}} \} ~|~ nil ] 
~ \& ~ 
(b)\  [\{\overrightarrow{I_{b}}\sigma\}  ~|~  \overrightarrow{M_b}\sigma 
] 
\,\&\, IK%
\notag
\end{align}

\noindent
where $(a, b,   \mathrm{1{-}1}) \in S$,
$(a) [
\overrightarrow{M_a},\{\overrightarrow{O_{a}}\}]$ is a role in $\Pone$,
$(b)[\{\overrightarrow{I_{b}}\},\overrightarrow{M_b}]$ is a role in $\Ptwo$,
 $\overrightarrow{O_{a}}, \overrightarrow{I_{b}}$
are two sequences of terms with variables,
$\overrightarrow{M_a},\overrightarrow{M_b}$ 
are two sequences of input and output messages,
$\overrightarrow{O_{a}} =_{E_\caP} \overrightarrow{I_{b}}\sigma$,
and only $SS$ and $IK$ are variables.
Since this rule was applied, 
there is a substitution $\rho$ such that 
$(a)\  [\overrightarrow{M_a}\rho, \{ \overrightarrow{O_{a}\rho} \} ~|~ nil ]$ 
and
$(b)\  [\{\overrightarrow{I_{b}}\sigma\rho\}  ~|~  \overrightarrow{M_b}\sigma\rho ] $
are strand instances in $St_1$.
But, by application of the \textit{trans} function,
there are strands
$(a)\  [\overrightarrow{M_a}\rho, \{ a \to b_1 \cdots b_{i-1}\ b\ b_i \cdots b_k ;; \textrm{1-1} ;; O^\sharp_{a}\rho \} ~|~ nil ]$ 
and
$(b)\  [\{a_1 \cdots a_{j-1}\ a\ a_j\cdots a_n  \to b ;; \textrm{1-1} ;;  I^\sharp_{b}\sigma\rho\}  ~|~  \overrightarrow{M_b}\sigma\rho ] $
in $St'_1$.
Now, since 
$\overrightarrow{O_{a}}\rho =_{E_\caP} \overrightarrow{I_{b}}\sigma\rho$,
the reversed version of 
Rule~\eqref{eq:one-to-one-forward} is applicable

\begin{align}
&
SS\,\&\, 
[ L ~|~ \{\textrm{a}\rightarrow \textrm{b}\ \textrm{R}  \ ;; \ \textrm{Mode} \ ;; \textrm{M}\}] 
 ~ \& ~
[ nil ~|~ \{\textrm{a} \ \textrm{R'} \rightarrow \textrm{b}\ ;; \ \textrm{Mode} \ ;; \textrm{M}\}, L'] 
\hspace{0ex}\,\&\, IK%
\notag\\
\rightarrow&
SS\,\&\, 
[ L, \{\textrm{a} \rightarrow \textrm{b} \ \textrm{R} \ ;; \ \textrm{Mode} \ ;; \textrm{M}\} ~|~ nil] 
~ \& ~
[ \{\textrm{a} \ \textrm{R'} \rightarrow \textrm{b} \ ;; \ \textrm{Mode} \ ;; \textrm{M}\} ~|~ L'] 
\hspace{0ex}\,\&\, IK%
\notag
\end{align}

\noindent
where $SS$, $L$, a, b, R, Mode, M, R', $L'$, $IK$ are variables.

\item[(c)]
A rule in the set~\eqref{eq:one-to-one-forward-new-transf}
corresponds to an application of 
a rule in the set~\eqref{eq:one-to-*-generated} 
(introducing a new parent strand and composing it with an existing child strand).
 
\item[(d)]
A rule in the set~\eqref{eq:one-to-many-forward-transf}
corresponds to an application of 
Rule~\eqref{eq:one-to-many-forward}
(synchronizing the output parameters of the parent strand
with the already accepted input parameters of the child strand,
but without moving the bar in the parent strand).
\qed
\end{itemize}
\end{proof}

Finally, we can put everything together into the following result.

\begin{theorem}[Soundness and Completeness]\label{thm:shortDC}
Let $\caP_{1}$ and $\caP_{2}$ be two protocols
and $\SeqComp$ their composition, as defined in Section \ref{subsec:composition-formalization}. 
 Let $\RTdc$ be the   rewrite theory associated to composition via  synchronization messages defined above
 in     Section~\ref{sec:semantics-directComp}, 
 and let
$\RTcomp$ be the rewrite theory associated to the abstract protocol composition, as described in Section~\ref{sec:compex}.

Given a state $St$ valid according to 
 $\RTcomp$
 and
 an initial state $St_{ini}$ 
such that $\textit{trans}(St)=St'$ and $\textit{trans}(St_{ini})=St'_{ini}$,
then 
$St_{ini}$ 
is reachable from $St$
by backwards narrowing 
in $\RTcomp$
iff
 $St'_{ini}$ is reachable from $St'$
by backwards narrowing 
in $\RTdc$.
\end{theorem}

\begin{proof}
By successive applications of Lemma~\ref{oneStepDC}.
\qed
\end{proof}

In the case of sequential composition of $n$ protocols 
$\caP_1\ ;_{S_1} \caP_2 \ ;_{S_2} \caP_3 \ ;_{S_3} \ldots ;_{S_{n-2}} \caP_{n-1} \ ;_{S_{n-1}} \caP_n$
as described in Definition~\ref{def:NprotocolComp},
we can define a function $\textit{trans}_{S_1,S_2,S_3,\ldots,S_{n-1}}$ 
between states valid according to the rewrite theory 
$\caR_{\caP_1\ ;_{S_1} \caP_2 \ ;_{S_2} \caP_3 \ ;_{S_3} \ldots ;_{S_{n-2}} \caP_{n-1} \ ;_{S_{n-1}} \caP_n}$ and 
states valid according to the rewrite theory 
$\caR_{synch(\caP_1\ ;_{S_1} \caP_2 \ ;_{S_2} \caP_3 \ ;_{S_3} \ldots ;_{S_{n-2}} \caP_{n-1} \ ;_{S_{n-1}} \caP_n)}$
with the only requirement that 
the role names of a protocol $\caP_i$ have to be different from the role names of all other protocols
$\caP_j$, $j \neq i$.
This requirement ensures that each strand instance can be easily associated to one of the protocols;
otherwise we may have a strand instance being associated to several protocol states.
We are working on relaxing this condition, perhaps via use of role adapters (Section~\ref{sec:adapters}).

\section{Composition via Protocol Transformation}\label{sec:comparison}\label{sec:transfComposition}

In this section we describe our previous approach to composition using protocol transformation,
presented
 in \citep{EMMS10,compositionTR2010}.  This section provides background for Section \ref{sec:experiments-composition}, in which the performance of
 composition via synchronization messages is compared with its predecessor.

 In \citep{EMMS10}, we presented an approach for protocol composition  
 where we defined a notion of sequential protocol composition
 slightly different from the one presented in Section~\ref{sec:mnpa-comp}  and
 the transition rules associated to such a composition. We did not
 implement those transition rules in the Maude-NPA.  Instead, we defined
 a protocol transformation that achieved the same effect
 using the existing Maude-NPA tool.  Proofs of soundness and
 completeness of the protocol transformation for the transition rules
 of \citep{EMMS10} were provided in \citep{compositionTR2010}.

 However, when experimenting with actual protocol composition
 examples, we realized that such a protocol composition and its
 semantics were quite complex and produced too many transition
 rules for a concrete protocol composition.  This led us to refine such
 protocol composition and its transition rules in the considerably
 simpler form now presented in Section~\ref{sec:mnpa-comp}. 
 Besides being simpler, it has also a more
 effective protocol composition semantics, more suitable for
 implementation. We then investigated two routes to obtaining
 a Maude-NPA implementation of the simpler composition
 notion and it semantics of Section~\ref{sec:mnpa-comp}:
\begin{enumerate}
\item the more direct route based on synchronization messages
presented in Section~\ref{subsec:synchro}; and

\item the older route from \citep{EMMS10}
 based on \emph{protocol transformation}, but now 
according to the new composition notion and associated semantics
of Section~\ref{sec:mnpa-comp}. 
\end{enumerate} 

This was then used as a basis to compare more carefully which of these
two possible implementation routes would be the best.
To begin with, we wanted to prove that both (1) and (2)
above provided \emph{correct} implementations.
The correctness of the synchronization-based route of (1)
has been proved in Section~\ref{sec:composition-SoundnessCompleteness}.
Similarly, in analogy with
\citep{EMMS10,compositionTR2010},
the  redefined and adapted notion of protocol transformation
in (2) has been proved \emph{correct} in  \citep{tesis-sonia} with respect to the new
protocol composition semantics of Section~\ref{sec:mnpa-comp}. 
Once we were sure that both
implementation routes were correct, we proceeded to
compare their ease of use, simplicity, and performance
through concrete case studies.

The rest of this section briefly describes route (2),
based on the 
protocol transformation.
A more detailed comparison of ease of use, simplicity, and performance
between (1) and (2) is postponed until
Section~\ref{sec:experiments-composition}.

Given two protocols $\caP_{1}$ and $\caP_{2}$, its sequential
composition implemented via the redefined and adapted protocol
transformation in (2),
written $\Phi(\caP_{1} \ ;_S \caP_{2})$, 
is a single, composed protocol specification where:
%
%
\begin{enumerate}
\item Sorts, symbols, 
and equational properties of both protocols are put together into a single specification.
As explained in Footnote~\ref{footnote:clash-renaming} in Section~\ref{sec:maude-npa}, 
we allow shared items but require the user to solve any possible conflict. 
\item A new sort \sort{Param} is defined to denote
input and output parameters.
The sort \sort{Param} 
is disjoint from the sort \sort{Msg} used by the protocol 
in the honest and intruder strands to ensure that an intruder cannot fake a composition.


\item
For each composition 
$(a, b, \textit{MODE})$
with underlying substitution $\sigma$ such that $\overrightarrow{O_a} \congr{E_\caP} \overrightarrow{I_b}\sigma$,
we transform 
the input parameters $\{\overrightarrow{I_{b}}\}$
into an input message exchange of the form
$-(\overrightarrow{I_{b}})$, 
and
the output parameters $\{\overrightarrow{O_{a}}\}$
into an output message exchange of the form
$+(\overrightarrow{I_{b}}\sigma)$.
In order to avoid type conflicts, we use a \emph{dot}
for  concatenation within protocol composition exchange messages,
e.g. input parameters $\overrightarrow{I}=\{A,B,NA\}$ are transformed into
the sequence $\dot{I}=A\; .\; B\; .\; NA$. 
%
\item
Each composition 
is uniquely identified by using a composition
identifier (a variable of sort \sort{Fresh}).
Strands exchange such composition identifier by using
input/output messages of the form $role_{j}(r)$,
which make the role explicit.
The sort \sort{Role} 
of these messages
is disjoint from 
the sorts \sort{Param} and \sort{Msg}.
\begin{enumerate}
\item
In a one-to-one protocol composition, the child strand uniquely generates
a fresh variable that is added to the area of fresh identifiers
at the beginning of its strand specification. This fresh variable must be passed from the child to the parent before the parent generates its output parameters and sends them back  
again to the child.
What this simulates in practice is the uniqueness of the one-to-one composition, since the parent can generate a single such message.
\item
In a one-to-many protocol composition, the parent strand uniquely generates a fresh variable that is passed to each child. Since an (a priori) unbounded number of children will be composed with it, no reply to the fresh variable 
is expected by the parent from the children.
Note that all the children strands receive the same fresh variables from the parent.
\end{enumerate}
\end{enumerate}

\noindent
%
Let us illustrate this protocol transformation 
with our examples of protocol compositions.

\begin{example}

 The transformed strands of the one-to-one 
 protocol composition $\mathit{NSL} ;_S \mathit{DB}$
 of Example~\ref{ex:protComp-NSLDB} are as shown below:


\begin{small}
\begin{align}
  :: r :: 
 [ nil | &+(\textit{NSL.init}), \notag\\
         &
         +(pk(B,n(A,r) ; A)) , 
         -(pk(A, n(A,r) ; NB ; B )), 
         +(pk(B, NB)), \notag\\
         &
         -(\textit{DB.resp}(r\#)),
         +(\textit{NSL.init}(r\#)~.~A~.~B~.~n(A,r)) ] \&\notag\\
:: r ::
[ nil | &+(\textit{NSL.resp}), \notag\\
        &
        -(pk(B,NA ; A)), 
        +(pk(A, NA ; n(B,r) ; B)), 
        -(pk(B,n(B,r))),\notag\\
        &
        -(\textit{DB.init}(r\#)),
        +(\textit{NSL.resp}(r\#)~.~A~.~B~.~NA ) ] \&\notag\\
:: r', r\# ::
[ nil | &
         +(\textit{DB.init}(r\#)),
        -(\textit{NSL.resp}(r\#)\ .\ A\ .\ B\ .\ NA  ), \notag\\
        &
        +(n(B,r')), 
        -(NA \oplus n(B,r')) ] \&\notag\\
:: r\# :: 
[ nil | &
         +(\textit{DB.resp}(r\#)),
        -(\textit{NSL.init}(r\#)\ .\ A\ .\ B\ .\ NA ), \notag\\
        &
        -(N), 
        +( NA \oplus N), nil ] \notag
\end{align}
\end{small}

\noindent

 The transformed strands of the 
 protocol composition $\mathit{NSL} ;_S \mathit{KD}$
 of Example~\ref{ex:protComp-NSLKD} are as shown below:

%

\begin{small}
\begin{align}
:: r , r\# :: 
[ nil | &+(\textit{NSL.init}),\notag\\
        &
        +(pk(B,n(A,r) ; A)) , 
        -(pk(A, n(A,r) ; NB ; B )), 
        +(pk(B, NB)), \notag\\
        &
        +(\textit{NSL.init}(r\#)\ .\ A\ .\ B\ .\ h(n(A,r) , NB) ) ] \&\notag\\
:: r , r\# ::
[ nil | &
         +(\textit{NSL.resp}),\notag\\
        &
        -(pk(B,NA ; A)), 
        +(pk(A, NA ; n(B,r) ; B)), 
        -(pk(B,n(B,r))),\notag\\
        &
        +(\textit{NSL.resp}(r\#)\ .\ B\ .\ A\ .\ h(NA , n(B,r))) ] \&\notag\\
:: r' ::
[ nil | &+(\textit{KD.init}), 
        -(\textit{RO1}\ .\ C\ .\ D\ .\ K ), \notag\\
        &
        +(e(K, skey(C, r'))) , 
        -(e(K, skey(C, r') ; N)), 
        +(e(K, N))] \&\notag\\
:: r' :: 
[ nil | &+(\textit{KD.resp}),
        -(\textit{RO2}\ .\ C\ .\ D\ .\ K ), \notag\\
        &
        -(e(K, SK)), 
        +(e(K, SK ; n(C,r'))), 
        -(e(K, n(C,r'))) ]     \&\notag
\end{align}
\end{small}

\noindent
where \textit{RO1} and \textit{RO2} are variables of sort \sort{Role}.
\end{example}

\section{Pragmatic and Experimental Evaluation}\label{sec:experiments-composition}

In this section we 
further explore  composition via protocol transformation versus composition via synchronization messages comparing them for ease of use and simplicity.
Furthermore, we
present some experimental results about the performance of the two approaches.
First, in Section~\ref{sec:experiments-NSLDB} we show the attack for the NSL-DB explained in Section~\ref{sec:examples-NSLDB}.
 Then we fix the NSL-DB protocol using a hash function,  as explained in Section~\ref{sec:examples-NSLDB}, and show that the protocol
is verified as secure by our tool, i.e., the search space is finite and no attack is  found.
Moreover, in Section~\ref{sec:experiments-NSLKD} we show that the NSL-KD protocol  presented in Section~\ref{sec:examples-NSLKD} is also verified as secure by the Maude-NPA. 
Each time 
 we show a protocol secure, we also show that a regular execution  can be performed, proving that the search space is not empty a priori;
  however, these regular execution proofs have not been included in this paper, though they are available online (see below).

Here, the reader can see that the attack state patterns 
associated to the transformed protocol are more complex and hence more
error prone when they have to be specified than 
the
attack state patterns for  composition via  synchronization messages, since the introduction of fresh variables for protocol composition has to be done manually. Also, the attack state patterns look more artificial in the protocol transformation
because of the back and forth messages.

In Section~\ref{sec:experiments-comparison} we provide more details of the experiments and compare  the results  
obtained using  both techniques.
All the experiments, including the source Maude-NPA files and the generated outputs, can be found at: %
{\small \url{http://www.dsic.upv.es/~sescobar/Maude-NPA/composition.html}}

\subsection{The NSL-DB Protocol}\label{sec:experiments-NSLDB}

We start with the NSL-DB protocol composition. 
As explained in Section~\ref{sec:examples-NSLDB}, this protocol 
has an attack in which  the honest principal $B$ thinks that he has heard from a principal 
$D$ (who may or may not be honest), but who has actually heard from an honest principal $A$.  
This covers, for example, the case in which $D$ is dishonest, and tries to pass on an honest principal's 
authenticated response as his own.
This attack is represented in Maude-NPA by an attack state pattern,
according to the protocol specification of Example~\ref{ex:example-NSLDB},
where:
(i) the first strand is Alice talking to some principal $C$ acting as NSL initiator and connecting
to a DB responder,
(ii) the second strand is Bob taking to some principal $D$ acting as DB initiator and receiving data
from NSL responder,
and
(iii)
we include disequality constraints for principal names, namely $a \neq D$ and $C \neq b$.

More specifically, the attack state pattern using  the protocol transformation technique 
is as follows:

\begin{small}
\begin{align}
  :: r :: 
   [ nil, &+(\textit{NSL.init}), \notag\\
          &
          +(pk(C,n(a,r) ; a)), 
          -(pk(a, n(a,r) ; NC ; C )), 
          +(pk(C, NC)) \notag\\
 \mid        \ &-(\textit{DB.resp}(r\#1)), 
          +(\textit{NSL.init}(r\#1)\ .\ a\ .\ C\ .\ n(a,r))] \ \&  &(\textit{NSL-DB-a0-PT})\notag\\
   :: r', r\#2 ::
  [ nil, &+(\textit{DB.init}(r\#2)),
         -(\textit{NSL.resp}(r\#2)\ .\ D \ .\ b\ .\ n(a,r) ), \notag\\
         &
         +(n(b,r')), 
         -(n(b,r') \oplus  n(a,r)) 
         \mid\ 
         nil ]  \notag\\
 &  \& ((a \neq D) , (C \neq b))\notag
\end{align}
\end{small} 

\noindent
And the backwards search from 
this attack pattern 
does not terminate\footnote{In [15] we reported termination, but this turned out to be a result of a bug in Maude-NPA's management of disequality constraints, which has since been corrected. The development of new semantics and implementation helped us to discover this bug.} due to a state space explosion, and
no initial state is found up to the depth reached by the analysis.

In protocol composition via  synchronization messages
the attack state pattern is as shown below:

\begin{small}
\begin{align}
 :: r :: 
  [ nil, &+(pk(C,n(a,r) ; a)), 
         -(pk(a, n(a,r) ; NC ; C)), 
         +(pk(C, NC)) \notag\\
& \mid    \  \{\textit{NSL.init} \to \textit{DB.resp}\ ;;\ \textrm{1-1}\ ;;\ (a ; C ; n(a,r))\}] \&   &(\textit{NSL-DB-a0-SM})\notag\\
  :: r' ::
  [ nil, &\ \{\textit{NSL.resp}\to \textit{DB.init}\ ;;\ \textrm{1-1}\ ;;\ (D ; b ; n(a,r))\},               \notag\\
         &+(n(b,r')), 
         -(n(a,r) \oplus n(b,r'))  \mid nil]\notag\\
  &\& (a \neq D , C \neq b)\notag
\end{align}
\end{small}


\noindent
The backwards search from 
this 
attack state
using  composition via  synchronization messages finds 
 an initial state from which it is reachable, and thus demonstrates a distance hijacking attack.  
The exchange of messages of this attack is as explained in Section~\ref{sec:examples-NSLDB}.

We then considered other  attacks similar to the distance hijacking attack which however produced a smaller search space.
In the following attack, we asked whether it is possible for an attacker to use an initiator $A$'s nonce to participate in the distance-bounding
part of the protocol without Alice having completed the corresponding NSL strand. 
The attack state is given below (note the different position of the vertical bars w.r.t. attack state
\textit{NSL-DB-a0-PT});

\begin{small}
\begin{align}
  :: r :: 
   [ nil, &+(\textit{NSL.init}), 
          +(pk(C,n(a,r) ; a)) \notag\\
        | \,
          &-(pk(a, n(a,r) ; NC ; C )), 
          +(pk(C, NC)), \notag\\
          &-(\textit{DB.resp}(r\#1)),
          +(\textit{NSL.init}(r\#1)\ .\ a\ .\ C\ .\ n(a,r)) ]  \&&(\textit{NSL-DB-a1-PT})\notag\\
   :: r', r\#2 ::
  [ nil, &+(\textit{DB.init}(r\#2)),
         -(\textit{NSL.resp}(r\#2)\ .\ D\  .\ b\ .\ n(a,r) ), \notag\\
         &+(n(b,r')), 
         -(n(b,r') \oplus  n(a,r)) | nil ]  \notag\\
  & \& (a \neq D , C \neq b)\notag
\end{align}
\end{small} 

\noindent
 This, besides being simpler, required only that
the bar move one step forward in the NSL strand, and produced a smaller search space in which the protocol transformation version was able to find an attack, and to terminate on the corrected version of the protocol, giving us a better opportunity compare the performance of the two approaches.   
The same result is obtained for the attack pattern \textit{NSL-DB-a1-SM} but we do not include it here.

As explained in Section~\ref{sec:examples-NSLDB}, the distance hijacking attack can be avoided using a hash function. 
The previous property for the NSL-DB is specified in the new version of the protocol with the following attack state pattern
using the protocol transformation:

\begin{small}
\begin{align}
  :: r :: 
   [ nil, &+(\textit{NSL.init}), \notag\\
          &+(pk(C,n(a,r) ; a)), 
          -(pk(a, n(a,r) ; NC ; C )), 
          +(pk(C, NC)) \notag\\
 \mid        \ &-(\textit{DB.resp}(r\#1)),
          +(\textit{NSL.init}(r\#1)\ .\ a\ .\ C\ .\ n(a,r))] \ \&  &(\textit{NSL-DB-a0-fix-PT})\notag\\
   :: r', r\#2 ::
  [ nil, &+(\textit{DB.init}(r\#2)),
         -(\textit{NSL.resp}(r\#2)\ .\ D \ .\ b\ .\ n(a,r) ), \notag\\
         &+(n(b,r')), 
         -(n(b,r') \oplus  h(D,n(a,r))) 
         \mid\ 
         nil ]  \notag\\
 &  \& ((a \neq D) , (C \neq b))\notag
\end{align}
\end{small} 

\noindent
However, as in the case NSL-DB protocol, the analysis  using the protocol transformation
does not terminate due to state space explosion and, thus, the security of the protocol 
for this 
attack state pattern cannot be proved.

The distance hijacking attack
via  synchronization messages is as follows:

\begin{small}
\begin{align}
 :: r :: 
  [ nil, &+(pk(C,n(a,r) ; a)), 
         -(pk(a, n(a,r) ; NC ; C)), 
         +(pk(C, NC)) \notag\\
\mid         & \ \{\textit{NSL.init} \to \textit{DB.resp}\ ;;\ \textrm{1-1}\ ;;\ (a ; C ; n(a,r))\}] \&   &(\textit{NSL-DB-a0-fix-SM})\notag\\
  :: r' ::
  [ nil, & \ \{\textit{NSL.resp}\to \textit{DB.init}\ ;;\ \textrm{1-1}\ ;;\ (D ; b ; n(a,r))\},               \notag\\
         &+(n(b,r')), 
         -(h(D,n(a,r)) \oplus (b,r'))  \mid nil]\notag\\
  &\& (a \neq D , C \neq b)\notag
\end{align}
\end{small}

\noindent
The analysis of this  protocol composition
using the composition via  synchronization messages,
 terminates finding no attack (see Table~\ref{table:experiments-composition}). 
 Thus, the attack state is unreachable. 

Therefore, we proceed in a similar way  as  we did before and provide an attack pattern
with an earlier position of the vertical bar:

\begin{small}
\begin{align}
  :: r :: 
   [ nil, &+(\textit{NSL.init}), 
          +(pk(C,n(a,r) ; a)) \notag\\
        |
          &-(pk(a, n(a,r) ; NC ; C )), 
          +(pk(C, NC)), \notag\\
          &-(\textit{DB.resp}(r\#1)),
          +(\textit{NSL.init}(r\#1)\ .\ a\ .\ C\ .\ n(a,r)) ]  \&&(\textit{NSL-DB-a1-fix-PT})\notag\\
   :: r', r\#2 ::
  [ nil, &+(\textit{DB.init}(r\#2)),
         -(\textit{NSL.resp}(r\#2)\ .\ D\  .\ b\ .\ n(a,r) ), \notag\\
         &+(n(b,r')), 
         -(n(b,r') \oplus h(D, n(a,r))) | nil ]  \notag\\
  & \& (a \neq D , C \neq b)\notag
\end{align}
\end{small}

 \noindent
 The analysis of the protocol using the protocol transformation
 terminates, finding no initial state from which this more specific
 attack state pattern is reachable. 
 %
 %
 The same result is obtained for the attack pattern \textit{NSL-DB-a1-fix-SM} but we do not include it here.

\subsection{The NSL-KD Protocol}\label{sec:experiments-NSLKD}

For the NSL-KD protocol 
presented in Section~\ref{sec:examples-NSLKD}
we may wish  to guarantee that 
a dishonest principal is not able to learn the secret key of an honest principal. This 
property is represented by an attack state pattern, according to the protocol of Example~\ref{ex:example-NSLKD},
where 
the first strand is an initiator of the KD protocol generating the 
session key $skey(a,n(a,r'))$,
the second strand is a responder of the KD protocol using the same session key $skey(a,n(a,r'))$,
and 
we ask whether the intruder can learn this session key by adding the fact $skey(a,n(a,r'))$
to the intruder knowledge.

More specifically, in the protocol transformation the attack state pattern
is of the following form:   

{\small
\begin{align}
  :: r' ::
    [ nil, &+(\textit{KD.init}),
           -(\textit{RO1}\ .\ a\ .\ b\ .\ K ),  \notag\\
           &
           +(e(K, skey(a, r'))) , 
           -(e(K, skey(a, r') ;  n(b,r))), 
           +(e(K,  n(b,r))) \mid nil] \& &(\textit{NSL-KD-PT})\notag\\
    :: r :: 
    [ nil, &+(\textit{KD.resp}),
           -(\textit{RO2}\ .\ b\ .\ a\ .\ K ), \notag\\
           &
           -(e(K, skey(a, r'))), 
           +(e(K, skey(a, r') ; n(b,r))), 
           -(e(K, n(b,r))) \mid nil ] \notag\\
 &   \& (\inI{skey(a, r')}) \notag
\end{align}
}

\noindent
whereas for the   composition via synchronization message is specified as follows:


{\small
\begin{align}
  :: r' ::
   [ nil, &\ \{\textit{NSL.init NSL.resp} \to \textit{KD.init} \ ;;\ \textrm{1-*}\ ;;\  (a ; b ; K) \}, \notag\\
          &
          +(e(K, skey(a, r'))), 
          -(e(K, skey(a,r') ; n(b,r))), 
          +(e(K, n(b,r)))  | nil] \&  &(\textit{NSL-KD-SM})\notag\\
   :: r :: 
   [ nil, &\ \{\textit{NSL.init NSL.resp} \to \textit{KD.resp}\ ;;\ \textrm{1-*}\ ;;\ (b ; a ;  K) \}, \notag\\
          &
          -(e(K, skey(a, r'))), 
          +(e(K, skey(a, r') ; n(b,r))), 
          -(e(K, n(b,r))) | nil ]\notag\\
   & \& (\inI{skey(a, r')})\notag
\end{align}
}


\noindent
Here again the reader can see that the attack state pattern for the transformed protocol 
lacks some useful information about what is really happening, since we have two strands, each participating in different protocol composition, but no indication of what the possible compositions are.  However, the attack state pattern for the composition via  synchronization messages clearly shows that the two different one-to-many compositions that are possible for each strand.

In this case, the desired property is satisfied by the NSL-KD, since the
analysis terminates using both the protocol transformation and the  composition via synchronization messages techniques,
finding no initial state for the attack state pattern described above.

\subsection{Performance Comparison}\label{sec:experiments-comparison}

 In this section we show in detail the results of the experiments 
presented in Sections~\ref{sec:experiments-NSLDB} and
\ref{sec:experiments-NSLKD}. 
 Table~\ref{table:experiments-composition} gathers the results
 of the analysis of these  protocol compositions, i.e.,
 (i) the composition of the NSL and DB protocols (NSL-DB),
 (ii) the composition of the NSL and the fixed version of the DB protocol (NSL-DB-fix),
 and
 (iii) the composition of the NSL and the KD protocols (NSL-KD).
Note that for the NSL-DB and NSL-DB-fix protocols we consider 
the two attack state patterns shown above: the more generic, denoted as ``a0'', e.g. NSL-DB-a0;
and the more specific, denoted as ``a1'', e.g. NSL-DB-a1.
%
%
 %
 For each protocol composition we provide the following
 information. 
 For each technique, i.e.,   protocol transformation
 and composition via  synchronization messages (referred as composition via SM in the table header), 
 the column ``Secure?''
 shows whether 
 the technique sucessfully proved
 the protocol composition is secure, i.e.
 Maude-NPA   generated a finite search space finding
 no attacks, or insecure, i.e, Maude-NPA found an attack.
 When Maude-NPA did not obtain a definite result,
 i.e., when the analysis did not terminate (e.g. because of an state space explosion) 
 and no  initial state was found up
 to the depth reached by the analysis, we write ``?" in this column.
 The column ``Finite?'' indicates  whether Maude-NPA generated a finite state search space 
 or not, i.e. whether the analysis of such protocol composition
 terminated or not.
The   column ``Depth''  provides the  depth of the analysis, 
 i.e., the number of  reachability steps performed by Maude-NPA
 until: (i) it generates a finite search space with no attacks in 
 the case of a secure composition,
  (ii) it finds the attack in the case of an insecure composition, or
  (iii) the analysis finished before obtaining a definite result;
  whereas
the column ``States'' shows  the total number of states 
generated during the analysis up to the indicated depth.
For the composition via  synchronization messages, the column ``SM / PT''
shows  the 
state space reduction as the number of states explored by the 
synchronization messages method (SM) divided by the number of states explored by the
protocol transformation  method (PT).
When Maude-NPA did not obtain concluding results using the protocol transformation technique
 we write ``-'' in this column.
In the case of the simpler attack for the NSL-DB-fix protocol (attack NSL-DB-fix-a1 in   Table~\ref{table:experiments-composition}), marked with an *,
we considered only the number of states generated until the first initial state was found with
both techniques, since Maude-NPA could not generate a finite search space 
in the protocol transformation approach.

\begin{table}[t]
{\small

\centering

\begin{tabular}{|c|c|c|c|c||c|c|c|c|c|c|c|}
\hline
  
 & \multicolumn{4}{c||}{Protocol Transformation}
 & \multicolumn{5}{c|}{Composition via SM}
 \\ 
 \hline
 Attack  & Secure? & Finite?   &  Depth & States &  Secure? &  Finite? &  Depth & States & SM / PT  
  \\
\hline

NSL-DB-a0 &  ? &  No &  10    &   3434  &  No &  Yes  & 16   & 1337 &   -   \\
\hline
NSL-DB-a1 &  No  & No &   16 &   1529 &  No     &  Yes & 13 &  259   &  0.17* \\
\hline
NSL-DB-fix-a0 &  ?  &  No   &   10 & 2650     &   Yes & Yes    & 19    & 1690 & -   \\
\hline
NSL-DB-fix-a1 &  Yes  & Yes &  17   & 273     &  Yes &  Yes &   16  &  103  & 0.38 \\
\hline
NSL-KD &  Yes  &   Yes & 19   & 1486    &  Yes & Yes  &  16  &  652  & 0.44 \\
\hline

\end{tabular}

\caption{Experiments with sequential protocol compositions}
\label{table:experiments-composition}
} 


\end{table}

Regarding the execution time of the experiments, we note that we present these for the purpose of comparing
the composition times rather than as the best possible times that can be achieved using our methods.  The Maude programming language offers several levels of programming, including core Maude that provides the basic functionality of Maude, and the meta-level, in which Maude programmers can design new functionalities.  Core Maude has been carefully optimized, and hence programs in core Maude run faster than programs at the meta-level.  Our approach has been to first implement functionality at the meta-level, and then, when it is well understood, have it implemented in core Maude.  Many of the features we use in this analysis, including narrowing modulo equational are still implemented in the meta level although work is ongoing in moving them to core Maude. 

With this in mind, we present the execution times as follows.
In the case of attack NSL-DB-a0, the analysis using the protocol composition
 failed to complete after several days, whereas 
 using the composition via synchronization messages  it completed in a little under 9 hours.
For  attack NSL-DB-a1 using protocol transformation, the tool did not complete, and took almost 2 days to find an attack, while
 it took 1/2 hour to find the attack when using synchronization messages.
Attack NSL-DB-fix-a0 ran for several days without finishing when using protocol transformations,  while it completed after 6 days
when using synchronization messages.   
For  attack NSL-DB-fix-a1, the execution time was reduced from an hour and a half
when using protocol transformations to 35 minutes when using composition via synchronization
messages. For attack NSL-KD the tool took nine hours to complete using protocol transformations versus one and one-half hours using synchronization messages.

The reader may wonder why attack NSL-DB-fix-a0 took so much longer to complete than the other attacks, even for synchronization messages.  Although we have not yet investigated the reasons in detail, we believe that it is because that attack makes the most extensive use of narrowing modulo exclusive-or, which is the most expensive operation.

%
In summary,   protocol transformation fails to provide a definite result
about the security of two of our experiments, namely the analysis of the NSL-DB and NSL-DB-fixed 
protocol compositions for the distance hijacking attack state pattern, whereas 
this problem does not occur with the composition via  synchronization messages.
Morevoer, composition via  synchronization messages 
generates a finite state search space in all cases, whereas with protocol transformation 
this happens in only two cases.  Moreover, in the case in which both complete or both find an attack, so that it
is possible to compare performance directly, both state space size and time spent improved significantly for synchronization messages.  SM / PT state space size ratios ranged from 0.17 to 0.44.   Ratios for time spent were even more dramatic, ranging from 0.000868 to 0.389.  Although we should be careful about drawing too many conclusions for such a small number of experiments, we believe that it is safe to conclude that composition via synchronization messages offers a significant improvement in both space and time efficiency.

\section{Related Work, Lessons Learned, and Future Directions}\label{sec:conc}

 \subsection{Related Work}

Our work addresses a somewhat different problem than most existing work on cryptographic protocol
composition, which generally does not address model-checking. Indeed, to the best of our knowledge, most protocol analysis model-checking tools simply use
concatenation of protocol specifications to express sequential composition.   
However, we believe that the problem we are addressing is an important one that 
tackles a widely acknowledged source of protocol complexity.  
For example,
in the Internet Key Exchange Protocol \citep{ikev1}
there are sixteen different one-to-many parent-child compositions of Phase
One and Phase Two protocols.
The ability to synthesize compositions
automatically can  greatly simplify the specification and
analysis of protocols like these.

Now that we have a mechanism for synthesizing compositions, we are ready to revisit existing research on composing protocols and their properties
and determine how we could best make use of it in our framework. There have been two approaches to this problem.  One, called {\em nondestructive} composition in \citep{datta03mfps}, is to
concentrate on properties of protocols and conditions on them that guarantee that properties satisfied separately are not violated by the composition.  This is often (although not always) applied to parallel composition.
This is, for example, the approach taken by 
Gong and Syverson \citep{Gong-Syverson-98}, Guttman and Thayer \citep{GuttmanThayer00},   Cortier and Delaune \citep{CortierDelaune09},
Ciob\^ac\u{a}  and Cortier  \citep{CiobacaC10}, Gro{\ss} and M{\"o}dersheim  \citep{GrossM11},
%
and, in the  computational
model, Canetti's Universal Composability \citep{CanettiLOS02}.  The conditions in this case are usually ones that can be verified syntactically, so Maude-NPA, or any other model
checker would only be of use here to supply an experimental method for testing various hypotheses about  syntactic conditions. 

Of more interest to us is the research that addresses the compositionality of the protocol properties themselves, called {\em additive} composition in \citep{datta03mfps}.
This addresses the development of logical systems and tools such as CPL, PDL, and CPSA cited earlier in this paper, in which inference rules are provided
for deriving complex properties of a protocol from simpler ones.  Since these are pure logical systems, they necessarily start from very basic statements concerning,
for example, what a principal can derive when it receives a message.  But there is no reason why the properties of the component protocols could
not be derived using model checking, and then composed using the logic.  This would give us the benefits of both model checking
(for finding errors and debugging), and logical derivations (for building complex systems out of simple components), allowing to switch between
one and the other as needed. Indeed, we think that Maude-NPA is well positioned in this respect.  For example, the notion of state in
strand spaces that it uses is very similar to that used by PDL \citep{cerv05csfw}, and we have already developed a simple property language that allows us to translate
the ``shapes'' produced by CPSA into Maude-NPA attack state patterns.  The next step in our research will be to investigate this
 connection more closely from the
point of view of compositionality.

\subsection{Lessons Learned and Future Directions}\label{subsec:future}

Our work has also taught us much about the optimum strategies for extending Maude-NPA.  First of all, although it is desirable to be conservative when extending the syntax and semantics, this should be done in such a way that the resulting semantics reflects the extended functionality in a natural way.  Secondly, the use of message passing over the Dolev Yao channel is expensive computationally and can lead to state space explosion.  Thus it should be used only when the properties of the Dolev Yao channel
are actually needed.  Thirdly, the use of an abstract semantics which is not actually implemented can be very helpful in assisting
us to experiment with different implementation approaches  in the tool itself.  This allowed us to compare performance of different approaches while understanding their relationship to the abstract semantics.  Thus we could be sure that we were not giving up correctness in order to obtain better performance, and we could understand to what degree we were losing expressiveness. 

Finally, since 
the work of \citep{EMMS10} we have discovered that sequential protocol composition
is a key idea for several other applications in protocol specification such
as protocol branching, secure communication channels, group protocols and protocols
with global state memory.
We believe that these applications can also be supported in Maude-NPA
with extensions of the methods presented in this paper to support more expressive composition languages.
We have performed a preliminary study of these applications but we leave for future work 
a deeper investigation on these topics.

%
%
%

\bibliographystyle{plain}
\bibliography{tex}

\end{document}